\documentclass[runningheads,a4paper]{llncs}
\usepackage{amsmath}

\usepackage{times}
\usepackage{helvet}
\usepackage{courier}
\usepackage{etex}
\usepackage{bussproofs}
\usepackage{amssymb}
\usepackage{graphicx}
\usepackage{subfigure}
\usepackage{url}
\usepackage{qtree}
\usepackage{multirow}
\usepackage{proof}
\usepackage{color}
\usepackage{framed}
\usepackage{pgf,tikz}
\usepackage{tikz}
\usepackage{makecell}
\usepackage{inputenc}
\usepackage{pgfplots}
\usepackage{algorithmic}

\usetikzlibrary{positioning}
\usetikzlibrary{backgrounds}
\usetikzlibrary{arrows,automata,snakes}

\newcommand\nextstate[1]{\textsf{Next}({#1})}
\newcommand\lra{\longrightarrow}
\newcommand\A{\Rightarrow}
\newcommand\tool[1]{\textsf{#1}}
\newcommand\sctl{\tool{SCTLProV}}
\newcommand\sctlprov{\tool{S{CTL}ProV}}
\newcommand\sctlprovr{\textsf{SCTLProV$_R$}} 
\newcommand\sctlm{\textsf{SCTL($\cal M$)}}
\newcommand\verds{\tool{Verds}}
\newcommand\nusmv{\tool{NuSMV}}
\newcommand\nuxmv{\tool{NuXMV}}
\newcommand\couic[1]{}
\newcommand{\code}[1]{\texttt{#1}}
\newcommand\rsa{\rightsquigarrow}

\newcommand\T{\rule{0pt}{2.6ex}}       
\newcommand\B{\rule[-1.2ex]{0pt}{0pt}} 

\usetikzlibrary{shapes.symbols,shapes.geometric,shadows}
\tikzset{
  inputfile/.style={
    shape=tape,
    draw,
    fill=white,
    tape bend top=none}}
    
\newcommand{\CTLP}{CTL$_P$}

\newcommand\CTL{\textsf{CTL}}
\newcommand\SCTL{\textsf{SCTL}}

\newcommand\BMC{\textsf{BMC}}
\newcommand\QBF{\textsf{QBF}}
\newcommand\OCaml{\textsf{OCaml}}
\newcommand\BDD{\textsf{BDD}}

\def\ssucc{\mathrel{
		\scalebox{.6}[1.0]{$\succ$}\mkern-4.5mu\scalebox{.8}[1.0]{$\succ$}}}

\makeatletter
\newcommand{\figcaption}{\def\@captype{figure}\caption}
\newcommand{\tabcaption}{\def\@captype{table}\caption}
\makeatother

\begin{document}

\title{\SCTL: Towards Combining Model Checking and Proof Checking}
\author{Ying Jiang\inst{1}\and Jian Liu\inst{1,2} \and Gilles Dowek\inst{3}\and Kailiang Ji\inst{4}}

\institute{State Key Laboratory of Computer Science, Institute of Software, Chinese Academy of Sciences \\ \email{\{jy, liujian\}@ios.ac.cn}\and
	University of Chinese Academy of Sciences \and Inria and ENS de Cachan\\ 61, avenue du Pr\'{e}sident Wilson 94235 CACHAN Cedex, France \\ \email{gilles.dowek@ens-cachan.fr} \and
	Noah's Ark Lab, Shanghai Huawei Technologies Co., Ltd, Shanghai, China\\ \email{kailiang.ji2013@gmail.com} }

\couic{
\author{Ying Jiang}
\affiliation{State Key Laboratory of Computer Science, Institute of Software, Chinese Academy of Sciences,\\
	P.O. Box 8718, 100190 Beijing, China}
\author{Jian Liu}
\affiliation{State Key Laboratory of Computer Science, Institute of Software, Chinese Academy of Sciences;\\
University of Chinese Academy of Sciences,\\ P.O. Box 8718, 100190 Beijing, China}
\email{liujian@ios.ac.cn}
\author{Gilles Dowek}
\affiliation{Inria and ENS de Cachan, 61, avenue du Pr\'{e}sident Wilson 94235 Cachan cedex, Paris, France}
\author{Kailiang Ji}
\affiliation{Noah's Ark Lab, Shanghai Huawei Technologies Co., Ltd, Shanghai, China}


\shortauthors{Y. Jiang; J. Liu; G. Dowek; and K. Ji}
}

\maketitle

\begin{abstract}
Model checking and automated theorem proving are two pillars of
formal methods.  This paper investigates model checking from an
automated theorem proving perspective, aiming at combining the
expressiveness of automated theorem proving and the complete
automaticity of model checking. 
The focus of this paper is on the
verification of temporal logic properties of Kripke models.
The main contributions of this paper are: first the definition of an
extended computation tree logic that allows polyadic predicate
symbols, then a proof system for this logic, taking Kripke models as
parameters, then, the design of a proof-search algorithm for this
calculus and a new automated theorem prover to implement it.  The
verification process is completely automatic, and produces either a
counterexample when the property does not hold, or a certificate when
it does.  The experimental result compares well to existing
state-of-the-art tools on some benchmarks, including an application to
air traffic control and the design choices that lead to this
efficiency are discussed.

\end{abstract}

\section{Introduction}

Model checking \cite{CGP01,BouajjaniJNT00,BaierKatoen08} and automated
theorem proving \cite{Fitting96,Loveland78,Burel09} are two pillars of
formal methods.  They differ by the fact that model checking often
uses decidable logics, such as propositional modal logics, while
automated theorem proving mostly uses undecidable ones, such as
first-order logic.  Nevertheless, model checking and automated theorem
proving have a lot in common, in particular, both of them are often
based on a recursive decomposition of problems, through the
application of rules.

Links between model checking and automated theorem proving have been
investigated for long.  For instance, Bounded Model Checking (\BMC)
\cite{BCCZ99,Zhang14,BiereCCSZ03} is based a reduction of model
checking to satisfiability of boolean or quantified boolean
formulae.

This paper investigates model checking from an automated theorem
proving perspective, but instead of using a reduction, it directly
provides a proof system to solve model checking problems. This permits
to combine the expressiveness of automated theorem proving and the
complete automaticity of model checking.

The first contribution of this paper is to propose a slight extension
of \CTL{} \cite{EmersonC82,EmersonH85}, called \CTLP{}.  In this
extension, we may refer explicitly to states of the model.  The
proposition $P(s)$, for instance, expresses what is usually expressed
with the judgment $s \models P$.  Thus $P$ here is not an proposition
symbol, but a unary predicate symbol.  This transformation can be
compared to the introduction of adverbial phrases in natural
languages, where we can say not only ``The sky will be blue in the
future'' but also ``The sky will be blue on Monday''.  A proposition
such as $EX(P)(s)$ must then be written $EX_x (P(x))(s)$. Indeed, as
the symbol $P$ is now a unary predicate symbol, it must be applied to
a state variable, which is bound by the modality $EX$.
This allows to introduce polyadic predicate that do not only express
properties of states, but also relations between states.  For instance, we
can express the existence of a sequence of states $s = s_0, s_1, ...$
starting from $s$ such that for all $i$, $s_i \lra s_{i+1}$ and one
can buy a left shoe at some state $s_n$ and then the right shoe of the
same pair at a later state $s_p$. This property is expressed by the
formula $EF_x(EF_y(P(x, y))(x))(s)$.

The second contribution of this paper is to propose a proof system for
\CTLP{} in the style of a sequent calculus.  The proof search in
\SCTL{} coincides with checking the validity of a formula in a Kripke
model.  Using such a proof system has several advantages.  First, it
permits to give a certificate, a formal proof, for the property when
it succeeds.  Such a certificate can be verified by an independent
proof checker, increasing the confidence in the proved property, and
can also be combined with proofs built by other means.

Secondly, when the verification of the given property fails, it
permits to generate a counterexample as a proof of the negation 
of the formula,
instead of a sequence of states or trees labeled with states, as in
traditional model checkers.  In particular, when providing a
counterexample for a formula containing nested modalities, such as
$EG_x(EG_y(P(x, y))(x))(s)$, we need to provide a tree labeled with
states, in such a way that for each state $a$ labeling a leaf of this
tree, the formula $EG_y(P(a, y))(a))$ does not hold. That is for each
of these states, we need to provide another tree. As we shall see,
such a hierarchical tree can be represented as a proof of the formula
$AF_x(AF_y(\neg P(x, y))(x))(s)$.

Different proof systems for temporal logic have been proposed (see,
for instance,
\cite{EmersonH85,FisherDP01,GabbayP08,PnueliK02,Reynolds01,BruennlerLange}).
When designing such a proof system, one of the main issues is to
handle co-inductive modalities, for instance, asserting the existence
of an infinite sequence of which all elements satisfy some property.
It is tempting to reflect this infinite sequence as an infinite proof
and then use the finiteness of the model to prune the search-tree in a
proof search method. Instead, we use the finiteness of the model to
keep our proofs finite, like in the usual sequent calculus.  This is
the purpose of the \emph{merge} rules of \SCTL{} in Figure
\ref{sctl_rules}.

\SCTL{} is shown to be decidable, and proof search in this calculus
always terminates.

The third contribution of this paper is an implementation of a proof
search method for \SCTL.  Instead of translating the temporal formulae
to Quantified Boolean Formulae ({\QBF}s) \cite{Zhang14} or to the
format of an existing theorem prover \cite{Ji15}, we develop a new
automated theorem prover tailored for \SCTL{}, called \sctl{},
in the programming language \OCaml\footnote{\url{http://ocaml.org/}}. The
source code of \sctl{} is available
online\footnote{\url{https://github.com/terminatorlxj/SCTLProV}}.
Designing our own system gives us a lot of freedom to optimize it. For
example, the visited states are stored globally in order to avoid
visiting repeatedly states during the verification process of
\sctl{}. In addition, the set of visited states can be stored as a
Binary Decision Diagram (\BDD{}) in order to reduce space
occupation. These strategies are commonly used in traditional model
checkers, but cannot be realized in usual theorem provers like
\tool{iProver Modulo}. On the other hand, when formally verifying a
system, theorem provers usually output proof trees as a diagnosis of
the system, while in traditional model checkers, only sequences of
states representing the counterexample of properties can be produced.
Like usual theorem provers, \sctl{} produces proof trees when
verifying a system. Thus, when solving \CTL{} model checking problems,
\sctl{} can produce more instructive information than traditional
model checkers, and can use more optimization strategies than
traditional theorem provers.

To illustrate the efficiency of \sctl{}, we compare it with an
automated theorem prover \tool{iProver
  Modulo}\footnote{\url{http://www.ensiie.fr/~guillaume.burel/blackandwhite_iProverModulo.html.en}},
a \QBF{}-based bounded model checker
\verds{}\footnote{\url{http://lcs.ios.ac.cn/~zwh/verds/index.html}},
and two \BDD{}-based symbolic model checker
\nusmv{}\footnote{\url{http://nusmv.fbk.eu/}} and
\nuxmv{}\footnote{\url{https://nuxmv.fbk.eu/}} on several
benchmarks.  The experimental results show that \sctl{} compares
well with these four tools.

The efficiency of \sctl{} depends on the following design choices: the
first is that, unlike traditional symbolic model checkers or bounded
model checkers, \sctl{} searches states in a doubly on-the-fly
(both the transition relation and the formula are
unfolded on-the-fly)
\cite{BCG95,VergauwenL93} style. Thus, the state space is usually not
needed to be fully generated. This avoids enumerating unneeded states during
the verification procedure.  The second is that, unlike traditional
on-the-fly model checking algorithms for \CTL{}
\cite{BCG95,VergauwenL93}, our proof search algorithm is in
continuation-passing style \cite{Appel06}, in order to reduce stack
operations.

\sctl{} provides a more expressive input language than most traditional 
model checkers: it provides both readable notations for the
definition of data structures such as records or lists with unbounded
length, and arbitrary algorithms for the definitions of transition
rules and of properties. 

The rest of the paper is organized as follows.  In Section \ref{CTLP},
we introduce the logic system \CTLP{}.  In Section \ref{SCTL}, we
introduce the proof system \SCTL{}.  In Section \ref{implementation},
we describe the proof search algorithm for \SCTL{} and the prover
\sctl{}, which is an implementation of \SCTL{}.  In Section
\ref{sect:fairness}, we show the verification of properties under
fairness constraints in \SCTL{}.  In Section \ref{experiment}, we
compare, on several benchmarks, \sctl{} with \tool{iProver Modulo},
\verds{}, \nusmv{}, and \nuxmv{},
respectively. We also present an application of \sctl{} to 
model and analyze a concept of operations for air traffic control. 
\ref{app:detail:data} shows the details of the
experimental data with benchmark \#1, \#2, and \#3;
\ref{bench4:data:detail} shows the details of the experimental data
with benchmark \#4; \ref{appendix:proof:sound:complete} shows the
detailed proof of the soundness and completeness of the \SCTL{}
system; \ref{appendix:proof:corret:prfsearch} shows the detailed proof
of the correctness of the proof search method.

\section{\CTLP}
\label{CTLP}

In this section, we present the logic \CTLP{}$({\cal M})$
taking a Kripke model $\cal M$ as the parameter.  

\begin{definition}[Kripke model]
	A Kripke model $\cal M$ is given by
	\begin{itemize}
		\item a finite non-empty set $S$, whose elements are called states,
		\item a binary relation $\lra$ defined on $S$, such that for each $s$
		in $S$, there exists at least one $s'$ in $S$, such that $s\lra s'$,
		\item and a family of relations, each being a subset of $S^n$
for some natural number $n$.
	\end{itemize}
\end{definition}
We write $\textsf{Next}(s)$ for the set $\{s'\mid s\lra s'\}$ which is always finite. 
A $path$ is a finite or infinite sequence of states $s_0,...,s_n$ or $s_0,s_1,...$
such that for each $i$, if $s_i$ is not the last element of the
sequence, then $s_{i+1}\in \textsf{Next}(s_i)$. A $path$-$tree$ is a finite or
infinite tree labeled by states such that for each internal node
labeled by a state $s$, the children of this node are labeled by the
elements of $\textsf{Next}(s)$.

Properties of such a model are expressed in a language tailored for
this model that contains a constant for each state $s$, also
written $s$;
and a predicate symbol for each relation $P$, also written $P$.

The grammar of \CTLP{}($\cal M$) formulae is displayed below:
\begin{small}
	
	$$\phi \ :=
	\left\{\begin{array}{l}
	\top\ | \ \bot \ | \ P(t_1, ..., t_n)\  | \neg P(t_1, ..., t_n)\  | \ \phi  \wedge \phi \ |\ \phi \vee \phi \ | \\
	\ AX_x(\phi)(t)\ | \ EX_x(\phi)(t) \ | \ AF_x(\phi)(t) \ | \ EG_x(\phi)(t) \ |\\
	\ AR_{x,y}(\phi_1,\phi_2)(t)\ | \ EU_{x,y}(\phi_1,\phi_2)(t)
	\end{array}
	\right.$$
	
\end{small}
where  $x,y$ are variables, and each of $t$ and $t_1 \ldots t_n$ is either a constant or a variable. 

Note that in this language, modalities are applied to formulae and states, binding variables in these formulae.  
More explicitly, modalities $AX$, $EX$, $AF$, and $EG$ bind the variable $x$ in $\phi$,
and modalities $AR$ and $EU$ bind respectively the variable $x$ in $\phi_1$ and $y$ in $\phi_2$. 
Also, the negation is applied to atomic formulae only,
so, as usual, negations must be pushed inside the formulae. We use the
notation $(t/x)\phi$ for the substitution of $t$ for $x$ in $\phi$. As
usual, in presence of binders, substitution avoids captures.

The following abbreviations are used.
\begin{itemize}
\item $\phi_1\A \phi_2 \equiv \neg \phi_1 \vee \phi_2$,
\item $EF_x(\phi)(t) \equiv EU_{z,x}(\top, \phi)(t)$,
\item $ER_{x, y}(\phi_1,\phi_2)(t) \equiv EU_{y,z}(\phi_2,((z/x)\phi_1 \wedge
  (z/y)\phi_2))(t)\vee EG_y(\phi_2)(t)$, where $z$ is a variable that
  occurs neither in $\phi_1$ nor in $\phi_2$,
\item $AG_x(\phi)(t) \equiv \neg (EF_x(\neg \phi)(t))$,
\item $AU_{x,y}(\phi_1,\phi_2)(t) \equiv \neg (ER_{x,y}(\neg\phi_1,\neg\phi_2)(t))$.
\end{itemize}

Hereafter, a formula starting with one of the modalities $AF$, $EF$, $AU$ and $EU$ will be called an inductive formula; 
and a formula starting with one of the modalities $AR$, $ER$, $AG$ and $EG$ will be called a co-inductive formula.

\begin{definition}[Validity]
Let $\mathcal{M}$ be a model and $\phi$ be a closed formula, the
validity of a formula $\phi$ in the model $\mathcal{M}$ is defined by
induction on $\phi$ in Figure \ref{validity}.
\end{definition}

\begin{figure}[h]
\centering
\begin{tabular}{|l|}
\hline 
${\cal M}\models P(s_1,...,s_n)$, if $\langle s_1,...,s_n\rangle \in P$ with $P$ an $n$-ary relation on $\cal M$; 
\\ \hline
${\cal M}\models \neg P(s_1,...,s_n)$, if $\langle s_1,...,s_n\rangle \notin P$ with $P$ an $n$-ary relation on $\cal M$;
\\ \hline
 ${\cal M}\models \top$ is always the case;
 \\ \hline
 ${\cal M}\models \bot$ is never the case; 
 \\ \hline
 ${\cal M}\models \phi_1\wedge\phi_2$, if ${\cal M}\models \phi_1$ and ${\cal M}\models \phi_2$;
 \\ \hline
 ${\cal M}\models \phi_1\vee\phi_2$, if ${\cal M}\models \phi_1$ or ${\cal M}\models \phi_2$;
 \\ \hline
 ${\cal M}\models AX_x(\phi_1)(s)$, if for each state $s'$ in $\textsf{Next}(s)$,
  ${\cal M}\models (s'/x)\phi_1$;
  \\ \hline
 ${\cal M}\models EX_x(\phi_1)(s)$, if there exists a state $s'$ in
  $\textsf{Next}(s)$ such that ${\cal M}\models (s'/x)\phi_1$;
  \\ \hline
 ${\cal M}\models AF_x(\phi_1)(s)$, if 
there exists a finite tree $T$ such that $T$ has root $s$, 
for each internal node $s'$,\\ the children of this node\\ are labeled by the elements
of \nextstate{$s'$} and for each leaf $s'$, $\vdash (s'/x)\phi_1$;
  \\ \hline
 ${\cal M}\models EG_x(\phi_1)(s)$, if there exists an infinite path
  $s_0,s_1,...$ starting from $s$,\\ such that for all natural numbers
  $i$, ${\cal M}\models (s_i/x)\phi_1$;
  \\ \hline
 ${\cal M}\models AR_{x, y}(\phi_1,\phi_2)(s)$, 
  if there exists an
  possibly infinite tree such that 
  the root of this tree is $s$,\\
  for each internal node $s'$, the children of this node 
  are labeled by the elements of $\textsf{Next}(s')$,\\
  for each node $s'$, $\models (s'/y)\phi_2$ and
  for each leaf $s'$, $\models (s'/x)\phi_1$.
 \\ \hline
\end{tabular}
\caption{Validity of a formula in \CTLP\label{validity}}
\end{figure}

\begin{remark}
From the definition above, we obtain ${\cal M}\models EF_x(\phi)(s)$, if there
exists an infinite path $s_0, s_1,...$ starting from $s$ and a natural
number $j$ such that ${\cal M}\models (s_j/x)\phi$, etc.
\end{remark}

\begin{example}\label{example:uv}

  This example is motivated from the example presented in
  \cite{PartoviL14}, where the specification of the motion planning of
  multi-robot \cite{Craig89} system is characterized by \CTL{}
  formulae. The specification states that in a partitioned map, each
  robot starting from an initial section in the map will eventually
  move to its destination section; at the same time, each robot should
  avoid reaching some section along the movement steps.
	
  In our example, however, we focus on a ``spatial'' property (i.e., a
  property that characterize a relation between states) that can not
  be easily expressed in \CTL{}, but rather straightforward in \CTLP{}.
	
  Consider a special robot: an unmanned vehicle that is designed to
  move on the surface of a planet, which are partitioned into finite
  pieces of small areas. The unmanned vehicle moves from one area to
  another at a time, and the position of the unmanned vehicle is
  considered to be a state. Thus, the set of possible positions of the
  unmanned vehicle forms the set of the states, and the moves from one
  position to another form the transition relation. There is a very
  basic property that the design of the unmanned vehicle must hold:
  the unmanned vehicle must not stay in a small set of areas
  infinitely long, to be more precise, for a given distance $\sigma$,
  at any state $s$, the unmanned vehicle will eventually move to some
  state $s'$ such that the distance (not the number of moves) between
  $s$ and $s'$ is larger than $\sigma$. This property can be easily
  characterized by the \CTLP{} formula
  $AG_x(AF_y(D_{\sigma}(x,y))(x))(s_0)$ (Figure~\ref{fig:uv}), where
  $s_0$ is the landing position of the unmanned vehicle, i.e., the
  initial state; and atomic formula $D_{\sigma}(x,y)$ characterize the
  spatial property that the distance between state $x$ and state $y$
  is larger than $\sigma$.
	
  Such a temporal and spatial property cannot be elegantly expressed
  in traditional temporal logics, as there are no mechanisms to speak
  about specific states in the syntax of these logics. Even in the
  semantics, only one state is under consideration at a time, it is
  hard to express relationships between two states or among tuples of
  states.
\end{example}

\begin{figure}[h]
	\centering
	\includegraphics[width=7cm]{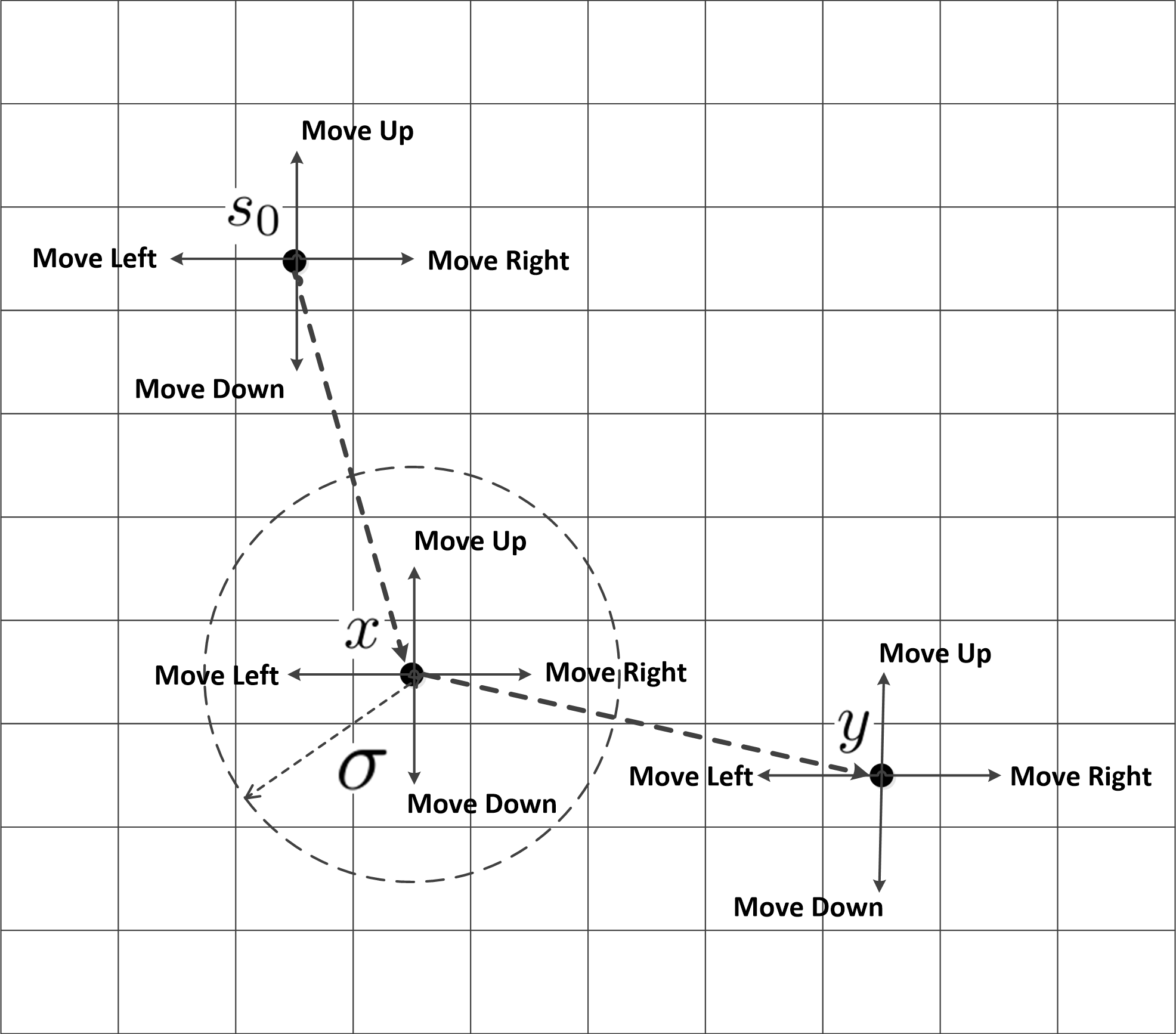}
	\caption{Possible positions of a unmanned vehicle.}\label{fig:uv}
\end{figure}

\section{\SCTL{}}
\label{SCTL}

In this section, we present \sctlm{}, a proof system for \CTLP{}$({\cal M})$.
Unlike the usual proof systems, where a formula is provable if and only
if it is valid in all models, a formula is provable in \sctlm{} if and
only if it is valid in the model ${\cal M}$.

First, consider the formula $AF_x(P(x))(s)$. 
This formula is valid if there exists a
finite tree $T$ whose root is labeled by $s$, such that the children
of an internal node labeled by a state $a$ are labeled by the elements
of $\textsf{Next}(a)$, and all the leaves are in $P$. Such a tree can
be called a \emph{proof} of the formula $AF_x(P(x))(s)$.

Now, consider $AF_x(AF_y(P(x, y))(x))(s)$ that contains
nested modalities. To justify the validity of this formula, one needs
to provide a tree whose root is labeled by $s$, where at each leaf $a$, the formula $AF_y(P(a,
y))(a)$ is valid. And to justify the validity of the formula
$AF_y(P(a, y))(a)$, one needs to provide other trees. These
hierarchical trees can be formalized with the proof rules
\begin{center}
$\infer[^{\mbox{\small $\mathbf{AF}$-{$\mathsf{R_1}$}}}]{\vdash AF_x(\phi)(s)}{\vdash (s/x)\phi}$\\
\vspace{5pt}
$\infer[^{\mbox{\small $\mathbf{AF}$-{$\mathsf{R_2}$}}}_{\{s_1,...s_n\}=\textsf{Next}(s)}]{\vdash AF_x(\phi)(s)}
{\vdash AF_x(\phi)(s_1) & \ldots & \vdash AF_x(\phi)(s_n)}$
\end{center}

{ \small 
\begin{example}\label{exp2}
Consider the model formed with the relation
\begin{center}
\centering
\begin{tikzpicture}
[scale=1.5,->,>=stealth',shorten >=1pt,auto,inner sep=2pt,semithick,bend angle=20]
  \tikzstyle{every state}=[draw=none] 
  \node[state] (1) at (0,0) {$a$};
  \node[state] (2) at (1,0.3) {$b$}; 
  \node[state] (3) at (1,-0.3){$c$}; 
  \node[state] (4) at (2,0) {$d$};
  
  \draw[right] (0.1,0.05) -- (0.9,0.25);
  \draw[right] (0.1,-0.05) -- (0.9,-0.25);
  \draw[right] (1.1,0.25) -- (1.9,0.05);
  \draw[right] (1.1,-0.25) -- (1.9,-0.05);
  \draw[->] (2.1,0.1) arc (-20:-325:-0.2);
\end{tikzpicture}
\end{center}
and the set $P=\{b, c\}$. A proof of the formula $AF_x(P(x))(a)$ is 
 
$$\infer[{\mbox{\small $\mathbf{AF}$-{$\mathsf{R_2}$}}}]{\vdash
  AF_x(P(x))(a)}{\infer[\mbox{\small $\mathbf{AF}$-{$\mathsf{R_1}$}}]{\vdash
    AF_x(P(x))(b)}{\infer[\mbox{\small atom-\textsf{R}}]{\vdash P(b)}{}} \quad\quad
  \infer[\mbox{$\mathbf{AF}$-{$\mathsf{R_1}$}}]{\vdash
    AF_x(P(x))(c)}{\infer[\mbox{\small atom-\textsf{R}}]{\vdash P(c)}{}}}$$
where besides the rules $\mathbf{AF}$-{$\mathsf{R_1}$} and $\mathbf{AF}$-{$\mathsf{R_2}$}, we use the rule 
\begin{center}
$\infer[^{\mbox{atom-\textsf{R}}}_{\langle s_1,...,s_n\rangle\in P}]{\vdash P(s_1,...,s_n)}{}$
\end{center}
\end{example}}

{ \small 
\begin{example}\label{exp3}
Consider the model formed with the same relation as in Example \ref{exp2}
and the set $Q=\{(b,d), (c,d)\}$. A proof of the formula 
$AF_x(AF_y(Q(x, y))(x))(a)$
is given in Figure \ref{afaf}.
\end{example}}

\begin{figure*}[h!]
\footnotesize
\centering
\noindent
\parbox{.78\textwidth}{\hspace*{-0.3cm}
$$\infer[{\mbox{\small $\mathbf{AF}$-{$\mathsf{R_2}$}}}]
        {\vdash AF_x(AF_y(Q(x, y))(x))(a)}
         {\infer[\mbox{$\mathbf{AF}$-{$\mathsf{R_1}$}}]
               {\vdash AF_x(AF_y(Q(x, y))(x))(b)}
               {\infer[{\mbox{\small $\mathbf{AF}$-{$\mathsf{R_2}$}}}]
                      {\vdash AF_y(Q(b, y))(b)}
                      {\infer[\mbox{\small $\mathbf{AF}$-{$\mathsf{R_1}$}}]
                             {\vdash AF_y(Q(b, y))(d)}
                             {\infer[\mbox{\small atom-\textsf{R}}]
                                    {Q(b,d)}
                                    {}
                              }
                      }
               }
         \quad\quad
         \infer[\mbox{$\mathbf{AF}$-{$\mathsf{R_1}$}}]
               {\vdash AF_x(AF_y(Q(x, y))(x))(b)}
               {\infer[{\mbox{\small $\mathbf{AF}$-{$\mathsf{R_2}$}}}]
                      {\vdash AF_y(Q(c, y))(c)}
                      {\infer[\mbox{\small $\mathbf{AF}$-{$\mathsf{R_1}$}}]
                             {\vdash AF_y(Q(c, y))(d)}
                             {\infer[\mbox{\small atom-\textsf{R}}]
                                    {Q(c,d)}
                                    {}
                              }
                      }
               }
        }$$
}
\caption{\label{afaf} a proof of $AF_x(AF_y(Q(x, y))(x))(a)$}
\end{figure*}

Note that \SCTL{} needs neither contraction rules nor multiplicative
$\vee$-\textsf{R} rules, because for each atomic formula $P$, either
$P$ is provable or $\neg P$ is.  Therefore the sequent $\vdash \neg P
\vee P$ is proved by proving either the sequent $\vdash \neg P$ or the
sequent $\vdash P$. 
As we have neither multiplicative $\vee$-\textsf{R} rules nor
structural rules, if we start with a sequent $\vdash \phi$, then each
sequent in the proof has one formula on the right of $\vdash$ and none
on the left.
So, as all sequents have the form $\vdash \phi$, the left rules and
the axiom rule can be dropped as well.  In other words, unlike the
usual sequent calculus and like Hilbert systems, \SCTL{} is tailored for
deduction, not for hypothetical deduction.

As the left-hand side of sequents is not used to record hypotheses, 
we will use it to record a different kind of information, that occur
in the case of co-inductive modalities, such as the modality $EG$. 

Indeed, the case of the co-inductive formula, for example $EG_x(P(x))(s)$, is
more complex than that of the inductive one, such as $AF_x(P(x))(s)$.
To justify its validity, one needs to provide an infinite sequence,
that is an infinite tree with only one branch, such that the root of
the tree is labeled by $s$, the child of a node labeled by a state $a$
is labeled by an element of $\textsf{Next}(a)$, and each node of the tree
verifies $P$.  However, as the model is finite, we can always restrict to
regular trees and use a finite representation of such trees. This
leads us to introduce a rule, called
$\mathbf{EG}$-\textsf{merge}, that permits to prove a sequent
of the form $\vdash EG_x(P(x))(s)$, provided such a sequent already
occurs lower in the proof. To make this rule local, we re-introduce
hypotheses $\Gamma$ to record part of the history of the proof. The
sequent have therefore the form $\Gamma\vdash\phi$, with a non empty
$\Gamma$ in this particular case only, and the
$\mathbf{EG}$-\textsf{merge} rule is then just an instance of
the axiom rule, that must be re-introduced in this particular case
only. The contexts of our sequents can be compared to the notion 
of history of \cite{BruennlerLange}, although our contexts are global 
while histories are attached to modalities.

The rules of \SCTL{} are depicted in Figure \ref{sctl_rules}.
\begin{figure}[h!]
\tiny
\centering
\noindent\framebox{\parbox{1.\textwidth}{\hspace*{-0.3cm}

$$
\infer[^{\mbox{atom-\textsf{R}}}_{\langle s_1,...,s_n\rangle \in P}]{\vdash P(s_1,...,s_n)}{}
\quad \quad
\infer[^{\mbox{$\neg$-\textsf{R}}}_{\langle s_1,...,s_n\rangle \notin P}]{\vdash \neg P(s_1,...,s_n)}{}
$$
\smallskip
$$
\infer[^{\mbox{$\top$-\textsf{R}}}]{\vdash \top}{}
\quad\quad
\infer[^{\mbox{$\wedge$-\textsf{R}}}]{\vdash \phi_1\wedge\phi_2}{\vdash \phi_1 & \vdash \phi_2}
\quad\quad
\infer[^{\mbox{$\vee$-{$\mathsf{R_1}$}}}]{\vdash \phi_1\vee\phi_2}{\vdash \phi_1}
\quad\quad
\infer[^{\mbox{$\vee$-{$\mathsf{R_2}$}}}]{\vdash \phi_1\vee\phi_2}{\vdash \phi_2}
$$
\smallskip
$$
\infer[^{\mbox{$\mathbf{EX}$-\textsf{R}}}_{s'\in \textsf{Next}(s)}]{\vdash EX_x(\phi)(s)}{\vdash (s'/x)\phi}
\quad\quad
\infer[^{\mbox{$\mathbf{AX}$-\textsf{R}}}_{\{s_1,...,s_n\}=\textsf{Next}(s)}]{\vdash AX_x(\phi)(s)}{\vdash (s_1/x)\phi & \ldots & \vdash (s_n/x)\phi}
$$
\smallskip
$$
\infer[^{\mbox{$\mathbf{AF}$-{$\mathsf{R_1}$}}}]{ \vdash AF_x(\phi)(s)}{\vdash (s/x)\phi}\quad\quad
\infer[^{\mbox{$\mathbf{AF}$-{$\mathsf{R_2}$}}}_{\{s_1,...,s_n\}=\textsf{Next}(s)}]{ \vdash AF_x(\phi)(s)}{\vdash AF_x(\phi)(s_1) & \ldots & \vdash AF_x(\phi)(s_n)}
$$
\smallskip
$$
  \infer[^{\mbox{$\mathbf{EG}$-\textsf{R}}}_{s' \in \textsf{Next}(s)}]{\Gamma \vdash EG_x(\phi)(s)}{\vdash (s/x)\phi & \Gamma,EG_x(\phi)(s)\vdash EG_x(\phi)(s')}\quad\quad
\infer[^{\mbox{$\mathbf{EG}$-\textsf{merge}}}_{EG_x(\phi)(s)\in \Gamma}]{\Gamma \vdash EG_x(\phi)(s)}{}
$$
\smallskip
$$
  \infer[^{\mbox{$\mathbf{AR}$-{$\mathsf{R_1}$}}}_{\{s_1,...,s_n\}=\textsf{{Next}}(s),\Gamma' = \Gamma,AR_{x,y}(\phi_1,\phi_2)(s)}]{\Gamma \vdash AR_{x,y}(\phi_1,\phi_2)(s)}{\vdash (s/y)\phi_2 & \Gamma'\vdash AR_{x, y}(\phi_1,\phi_2)(s_1) ~...~ \Gamma'\vdash AR_{x, y}(\phi_1,\phi_2)(s_n)}
$$
\smallskip

$$
\infer[^{\mbox{$\mathbf{AR}$-{$\mathsf{R_2}$}}}]{\Gamma \vdash AR_{x,y}(\phi_1,\phi_2)(s)}{\vdash (s/x)\phi_1 & \vdash (s/y)\phi_2}
\quad\quad
\infer[^{\mbox{$\mathbf{AR}$-\textsf{merge}}}_{AR_{x,y}(\phi_1,\phi_2)(s)\in \Gamma}]{\Gamma \vdash AR_{x,y}(\phi_1,\phi_2)(s)}{}
$$
\smallskip
$$
  \infer[^{\mbox{$\mathbf{EU}$-{$\mathsf{R_1}$}}}]{\vdash EU_{x,y}(\phi_1,\phi_2)(s)}{\vdash (s/y)\phi_2} 
\quad\quad
  \infer[^{\mbox{$\mathbf{EU}$-{$\mathsf{R_2}$}}}_{s'\in \textsf{Next}(s)}]{\vdash EU_{x,y}(\phi_1,\phi_2)(s)}{\vdash (s/x)\phi_1 & \vdash EU_{x,y}(\phi_1,\phi_2)(s')}
$$

}}
\caption{\label{sctl_rules} \sctlm{}}
\end{figure}

\begin{theorem}[Soundness and Completeness]\label{theom:sound:complete}
	If $\phi$ is closed, then 
	the sequent $\vdash \phi$ has a proof in \sctlm{} if and only if ${\cal M}\models \phi$ for the given Kripke model $\cal M$.
\end{theorem}
\begin{proof}
	The soundness and completeness are guaranteed by the finiteness of the Kripke model. The details are presented in 
\ref{appendix:proof:sound:complete}.
\end{proof}

\section{\textsf{SCTLP\MakeLowercase{ro}V}}\label{implementation}

In this section, the system \sctl{}\footnote{\url{https://github.com/terminatorlxj/SCTLProV}}, that is an implementation
of \SCTL{}, is presented and compared with other model checking tools.

\subsection{Implementation}

We develop, 
in the programming language \OCaml{}, 
a new automated theorem prover \sctl{} (Figure
\ref{workflow}) to implement \SCTL{}. 

\sctl{} reads and interprets an input file
containing a description of a Kripke model and a finite number of
formulae---the properties to be verified on the model. It searches for
a proof of these formulae and outputs a certificate (resp. True) when
the verification succeeds, and a counterexample, that is a proof of
the negation of the formula, (resp. False) when it does not.

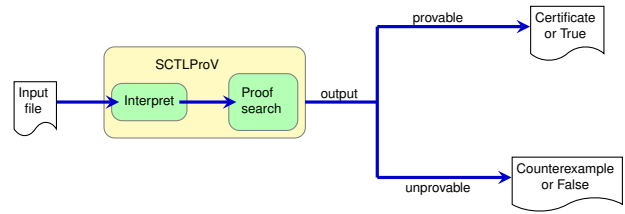
\begin{figure}[h]
	\centering
\makebox[0.50\textwidth][c]{%
\tiny
\begin{tikzpicture}[
outpt/.style={->,blue!80!black,very thick},
outpt1/.style={->,red!80!black,very thick},
>=stealth,
every node/.append style={align=center}]
 
\node [inputfile] (input) at (0,0) {
\begin{tabular}{@{}c}
\textsf{Input}\\\textsf{file}
\end{tabular}};
 
\node (preproc) at (1.5,0) {\textsf{Interpret}};
 
\draw[outpt](input)--(preproc);
 
\node (proofsearch) at (3.0,0) {
\begin{tabular}{@{}l}
\textsf{Proof}\\
\textsf{search}
\end{tabular}};
 
\draw[outpt](preproc)--(proofsearch);
 
\node [inputfile] (prooftree) at (7,1) {
\begin{tabular}{@{}c}
\textsf{Certificate} \\\textsf{or True}
\end{tabular}};
 
\node [inputfile] (result) at (7, -1) {
\begin{tabular}{@{}c}
\textsf{Counterexample} \\\textsf{or False}
\end{tabular}};
 
\node (sctlprov) at (2,0.5) {\sctl{}};
\node (output) at (4.0,0.1) {\textsf{output}};
\node (provable) at (5.3, 1.1) {\textsf{provable}};
\node (unprovable) at (5.3,-1.15) {\textsf{unprovable}};
 

\begin{pgfonlayer}{background}
\path (preproc.west |- sctlprov.north)+(-0.2,0.1) node (a) {};
\path (proofsearch.east |- proofsearch.south)+(+0.2,-0.2) node (c) {};
\path[fill=yellow!30,rounded corners, draw=black!50]
(a) rectangle (c);
\end{pgfonlayer}
 
\begin{pgfonlayer}{background}
\path (preproc.west |- preproc.north)+(-0.1,0.1) node (a) {};
\path (preproc.east |- preproc.south)+(+0.1,-0.1) node (c) {};
\path[fill=green!30,rounded corners, draw=black!50]
(a) rectangle (c);
\end{pgfonlayer}
 
\begin{pgfonlayer}{background}
\path (proofsearch.west |- proofsearch.north)+(-0.1,0.1) node (a) {};
\path (proofsearch.east |- proofsearch.south)+(+0.1,-0.1) node (c) {};
\path[fill=green!30,rounded corners, draw=black!50]
(a) rectangle (c);
\end{pgfonlayer}
 
\draw[-,very thick, blue!80!black](3.55, 0)--(4.5, 0);
\draw[-,very thick, blue!80!black](4.5, 0)--(4.5, 1);
\draw[-,very thick, blue!80!black](4.5, 0)--(4.5, -1);
\draw[outpt](4.5,-1)--(result);
\draw[outpt](4.5,1)--(prooftree);
\end{tikzpicture}
}
\caption{A general work flow of \sctl{}.}\label{workflow}
\end{figure}

The basic idea of the proof search procedure in \sctl{} is as
follows: first we give an order over the inference rules of \SCTL{}
with the same conclusion (if any) and, for each root under
consideration of a Continuation Passing Tree (Definition \ref{CPT}),
we give an order over the children of this node.  Then, to prove an
\SCTL{} sequent $\Gamma\vdash \phi$, we need to find an inference rule
of \SCTL{} such that this sequent matches the conclusion of the rule,
and then find successively a proof of each premise, according to the
given orders.  Thus, the proving procedure of sequent
$\Gamma\vdash \phi$ transforms into the proving procedure of all its
premises with some specific order.  
The major techniques used in this implementation are the use of 
continuations and of memorization. 

\subsubsection{The use of continuations}

One of the major techniques for the implementation of \SCTL{} is based
on the concept of continuation, usually used in compiling and
programming \cite{Appel06,Sestoft12}.  Basically, a continuation is an
explicit representation of ``the rest of the computation'', which will
happen next.

\begin{figure}[h!]
\scriptsize
\centering
\begin{tabular}{|ll|}
  \hline
  $\textsf{cpt}(\vdash\top, c_1, c_2)\rsa c_1$ & $\textsf{cpt}(\vdash\bot, c_1, c_2)\rsa c_2$
  \T\B\\ \hline
  \multicolumn{2}{|l|}{$\textsf{cpt}(\vdash P(s_1,...,s_n), c_1, c_2)\rsa c_1 \left[\langle s_1,...,s_n\rangle\in P\right]$}
  \T\B\\ \hline
  \multicolumn{2}{|l|}{$\textsf{cpt}(\vdash P(s_1,...,s_n), c_1, c_2)\rsa c_2 \left[\langle s_1,...,s_n\rangle\notin P\right]$}\\\hline
  
  \multicolumn{2}{|l|}{$\textsf{cpt}(\vdash \neg P(s_1,...,s_n), c_1, c_2)\rsa c_2 \left[\langle s_1,...,s_n\rangle\in P\right]$}
  \T\B\\ \hline
  
  \multicolumn{2}{|l|}{$\textsf{cpt}(\vdash \neg P(s_1,...,s_n), c_1, c_2)\rsa c_1 \left[\langle s_1,...,s_n\rangle\notin P\right]$}\\\hline
  
  \multicolumn{2}{|l|}{$\textsf{cpt}(\vdash \phi_1\wedge\phi_2, c_1, c_2)\rsa \textsf{cpt}(\vdash\phi_1, \textsf{cpt}(\vdash\phi_2, c_1, c_2), c_2)$}\T\B\\\hline
  
\multicolumn{2}{|l|}{ $\textsf{cpt}(\vdash\phi_1\vee\phi_2, c_1, c_2)\rsa \textsf{cpt}(\vdash\phi_1, c_1, \textsf{cpt}(\vdash\phi_2, c_1, c_2))$}\T\B\\ \hline
  
\multicolumn{2}{|l|}{$\textsf{cpt}(\vdash AX_x(\phi)(s), c_1, c_2)\rsa$}\\  
\multicolumn{2}{|l|}{$\textsf{cpt}(\vdash(s_1/x)\phi, \textsf{cpt}(\vdash(s_2/x)\phi, \textsf{cpt}(...\textsf{cpt}(\vdash(s_n/x)\phi, c_1, c_2),...,c_2), c_2), c_2)$}\\
\multicolumn{2}{|l|}{$\left[\{s_1,...,s_n\}=\textsf{Next}(s)\right]$}
\T\B\\ \hline

\multicolumn{2}{|l|}{$\textsf{cpt}(\vdash EX_x(\phi)(s), c_1, c_2)\rsa$}\\
\multicolumn{2}{|l|}{ $\textsf{cpt}(\vdash(s_1/x)\phi, c_1, \textsf{cpt}(\vdash(s_2/x)\phi, c_1, \textsf{cpt}(...\textsf{cpt}(\vdash(s_n/x)\phi, c_1, c_2)...)))$}\\
\multicolumn{2}{|l|}{$\left[\{s_1,...,s_n\}=\textsf{Next}(s)\right]$}
\T\B\\ \hline

\multicolumn{2}{|l|}{$\textsf{cpt}(\Gamma\vdash AF_x(\phi)(s), c_1, c_2)\rsa c_2 \left[AF_x(\phi)(s)\in \Gamma\right]~~~~~~$}\T\B\\ \hline
\multicolumn{2}{|l|}{$\textsf{cpt}(\Gamma\vdash AF_x(\phi)(s), c_1, c_2)\rsa$}\\
\multicolumn{2}{|l|}{$\textsf{cpt}(\vdash(s/x)\phi, c_1, \textsf{cpt}(\Gamma'\vdash AF_x(\phi)(s_1), \textsf{cpt}(...\textsf{cpt}(\Gamma'\vdash AF_x(\phi)(s_n), c_1, c_2)..., c_2), c_2))$}\\
\multicolumn{2}{|l|}{$\left[\{s_1,...,s_n\}=\textsf{Next}(s), AF_x(\phi)(s)\notin \Gamma, \textup{and}\; \Gamma'=\Gamma,AF_x(\phi)(s)\right]$}\B\\ \hline

\multicolumn{2}{|l|}{$\textsf{cpt}(\Gamma\vdash EG_x(\phi)(s), c_1, c_2)\rsa c_1 \left[EG_x(\phi)(s)\in \Gamma\right]~~~~~~$}\T\B\\ \hline
\multicolumn{2}{|l|}{$\textsf{cpt}(\Gamma\vdash EG_x(\phi)(s), c_1, c_2)\rsa$}\\
\multicolumn{2}{|l|}{$\textsf{cpt}(\vdash(s/x)\phi, \textsf{cpt}(\Gamma'\vdash EG_x(\phi)(s_1), c_1, \textsf{cpt}(...\textsf{cpt}(\Gamma'\vdash EG_x(\phi)(s_n), c_1, c_2)...)), c_2)$}\\
\multicolumn{2}{|l|}{$\left[\{s_1,...,s_n\}=\textsf{Next}(s), EG_x(\phi)(s)\notin \Gamma, \textup{and}\; \Gamma'=\Gamma, EG_x(\phi)(s)\right]$}\B\\ \hline

\multicolumn{2}{|l|}{$\textsf{cpt}(\Gamma\vdash AR_{x,y}(\phi_1,\phi_2)(s), c_1, c_2)\rsa c_1\left[(AR_{x,y}(\phi_1,\phi_2)(s)\in \Gamma\right]~~~~~~$}\T\B\\\hline
\multicolumn{2}{|l|}{$\textsf{cpt}(\Gamma\vdash AR_{x,y}(\phi_1,\phi_2)(s), c_1, c_2)\rsa$}\\
\multicolumn{2}{|l|}{$\textsf{cpt}(\vdash(s/y)\phi_2, \textsf{cpt}(\vdash (s/x)\phi_1, c_1, \textsf{cpt}(\Gamma'\vdash AR_{x,y}(\phi_1,\phi_2)(s_1), \textsf{cpt}(...\textsf{cpt}(\Gamma'\vdash AR_{x,y}(\phi_1,\phi_2)(s_n),$}\\
\multicolumn{2}{|l|}{$  c_1, c_2)..., c_2), c_2)), c_2)
\left[\{s_1,...,s_n\}=\textsf{Next}(s), AR_{x,y}(\phi_1,\phi_2)(s)\notin \Gamma, \textup{ and } \Gamma'=\Gamma,AR_{x,y}(\phi_1,\phi_2)(s)\right]$}\B\\ \hline

\multicolumn{2}{|l|}{$\textsf{cpt}(\Gamma\vdash EU_{x,y}(\phi_1,\phi_2)(s), c_1, c_2)\rsa c_2$ $\left[EU_{x,y}(\phi_1,\phi_2)(s)\in \Gamma\right]~~~~~~$}\T\B\\\hline
\multicolumn{2}{|l|}{$\textsf{cpt}(\Gamma\vdash EU_{x,y}(\phi_1,\phi_2)(s), c_1, c_2)\rsa$}\\
\multicolumn{2}{|l|}{$ \textsf{cpt}(\vdash(s/y)\phi_2, c_1, \textsf{cpt}(\vdash(s/x)\phi_1,\textsf{cpt}(\Gamma'\vdash EU_{x,y}(\phi_1,\phi_2)(s_1), c_1, \textsf{cpt}(...\textsf{cpt}(\Gamma'\vdash EU_{x,y}(\phi_1,\phi_2)(s_n),$}\\
\multicolumn{2}{|l|}{$ c_1, c_2)...)), c_2))\left[\{s_1,...,s_n\}=\textsf{Next}(s), EU_{x,y}(\phi_1,\phi_2)(s)\notin \Gamma, \textup{ and } \Gamma'=\Gamma, EU_{x,y}(\phi_1,\phi_2)(s)\right]$} \B\\ \hline
\end{tabular}
\caption{Rewritings over \textsf{CPT}s.}\label{fig:rewriting}
\end{figure}

\begin{definition}[Continuation Passing Tree]\label{CPT}
A Continuation Passing Tree (\textsf{CPT}) is a binary tree such that
\begin{itemize}
\item 
every leaf is labeled by either $\mathfrak{t}$ or $\mathfrak{f}$, 
where $\mathfrak{t}$ and $\mathfrak{f}$ are two different symbols;
\item 
every internal node is labeled by an \SCTL{} sequent. 
\end{itemize}
For each internal node in a \textsf{CPT}, the left subtree is called its $\mathfrak{t}$-continuation, 
and the right one its $\mathfrak{f}$-continuation. 
A \textsf{CPT} $c$ with an \SCTL{} sequent $\Gamma \vdash \phi$ as its root is often denoted by $\textsf{cpt}(\Gamma\vdash\phi, c_1, c_2)$, 
or visually by 
$$\Tree [.{$\Gamma\vdash\phi$} [.{$c_1$} ]  [.{$c_2$} ] ]$$ 
where  $c_1$ is the $\mathfrak{t}$-continuation of $c$, and $c_2$ the $\mathfrak{f}$-continuation.
\end{definition}

\textsf{CPT}s are evaluated to $\mathfrak{t}$ or $\mathfrak{f}$ using
the conditional rewrite rules presented in Figure \ref{fig:rewriting}
where conditions are put in brackets, which implement the rules of
\SCTL{}.

Note that there is no congruence rule to allow
the application of a rewrite rule to a subexpression of a \textsf{CPT}. So 
reductions always occur at the root of the \textsf{CPT}s.

The aim of the rewrite rules is to decide, for a given sequent $\Gamma
\vdash \phi$, if the \textsf{CPT} $\textsf{cpt}(\Gamma \vdash\phi, \mathfrak{t},
\mathfrak{f})$ reduces to $\mathfrak{t}$ or $\mathfrak{f}$. 
To do so,
we analyze the form of the formula $\phi$. If, for instance, it is 
$\vdash \phi_1 \wedge \phi_2$, we transform, using one of the
rewrite rules, the tree $\textsf{cpt}(\vdash \phi_1 \wedge \phi_2,
\mathfrak{t}, \mathfrak{f})$ into $\textsf{cpt}(\vdash \phi_1,
\textsf{cpt}(\vdash \phi_2, \mathfrak{t}, \mathfrak{f}), \mathfrak{f})$
expressing that if the attempt to prove $\vdash \phi_1$ succeeds
then we attempt to prove $\vdash \phi_2$, 
otherwise it just returns a negative result. 
The \textsf{CPT} $\textsf{cpt}(\vdash \phi_1, \textsf{cpt}(\vdash \phi_2, \mathfrak{t},
\mathfrak{f}), \mathfrak{f})$ is in turn transformed according to the
form of $\phi_1$.

\begin{proposition}[Termination]
$\textsf{cpt}(\vdash \phi, \mathfrak{t}, \mathfrak{f})$ always rewrites to $\mathfrak{t}$ or $\mathfrak{f}$ in finite many steps.
\end{proposition}
\begin{proof}
	Let $n$ be the cardinal of ${\cal M}$, let $|\phi|$ be the 
	size of $\phi$ defined in the usual way, and $|\Gamma|$
	be the length of a merge $\Gamma$. We define the size of a sequent $\Gamma \vdash 
	\phi$ as 
	$$|\Gamma \vdash \phi| = \langle|\phi|, (n - |\Gamma|)\rangle$$
	
	We then define the set of operators $F = \{\mathfrak{t}, \mathfrak{f}, \textsf{cpt}\}\cup \textsf{Seq}$, where $\textsf{Seq}$ is the set of sequents. The arity of $\textsf{cpt}$ is 3, while other elements in $F$ have arity 0. The partial ordering $\succ$ over $F$ is defined as follows:
	\begin{itemize}
		\item $\textsf{cpt} \succ \mathfrak{t}$;
		\item $\textsf{cpt} \succ \mathfrak{f}$;
		\item $\Gamma\vdash\phi \succ \textsf{cpt} $ for each sequent $\Gamma\vdash\phi$;
		\item $\Gamma\vdash \phi \succ \Gamma'\vdash\phi'$ iff $|\Gamma\vdash \phi| > |\Gamma'\vdash\phi'|$, where $>$ is the lexicographic ordering of natural numbers.
	\end{itemize}
	
	Let $T(F)$ be the set of terms constructed by operators in $F$, and $\succ^*$ be the recursive path ordering on $T(F)$ proposed by Dershowitz \cite{ND82}, which is restated here.
	
	The recursive path ordering $\succ^*$ on the set $T(F)$ of terms over $F$ is defined recursively as follows:
	$$s=f(s_1,...,s_m)\succ^*g(t_1,...,t_n)=t$$ iff one of the following cases holds:
	\begin{itemize}
		\item $f = g$, and $\{s_1,...,s_m\}\ssucc^*\{t_1,...,t_n\}$;
		\item $f \succ g$, and $\{s\}\ssucc^*\{t_1,...,t_n\}$;
		\item $f\not\succeq g$, and either $\{s_1,...,s_m\}\ssucc^*\{t\}$ or $\{s_1,...,s_m\}=\{t\}$;
	\end{itemize}
	where $\ssucc^*$ is the extension of $\succ^*$ to multisets. $\succ^*$ is well-founded since $\succ$ is well-founded \cite{ND82}. 
	
	To prove the termination of the rewriting system, it is sufficient to prove that for each rewriting $c\rsa c'$, $c\succ^*c'$.
	
	Then we analyze each case of $c\rsa c'$ in the rewriting system.
	Assume that $c$ is of the form $\textsf{cpt}(\Gamma\vdash\phi,c_1,c_2)$. 
	\begin{itemize}
		\item If $\Gamma\vdash\phi= $ $\vdash\top, \vdash\bot, \vdash P(s_1,...,s_m), \vdash P(s_1,...,s_m), \vdash \neg P(s_1,...,s_m)$, or $\vdash\neg P(s_1,...,s_m)$, and given that $\{\Gamma\vdash\phi, c_1, c_2\}\ssucc^*\{c_1\}$ and $\{\Gamma\vdash\phi, c_1, c_2\}\ssucc^*\{c_2\}$, then $c\succ^*c'$ by the definition of recursive path ordering on $T(F)$;
		\item If $\Gamma\vdash\phi =$ $\vdash \phi_1\wedge\phi_2$, and given that $ \{\vdash\phi_1\wedge\phi_2, c_1, c_2\}\ssucc^*\{\vdash\phi_1, \textsf{cpt}(\vdash\phi_2, c_1, c_2), c_2\}$, we have $c\succ^*c'$ by the definition of the recursive path ordering on $T(F)$;
		\item If $\Gamma\vdash\phi =$ $\vdash \phi_1\vee\phi_2$, and given that $\{\vdash\phi_1\vee\phi_2, c_1, c_2\}\ssucc^*\{\vdash\phi_1, c_1, \textsf{cpt}(\vdash\phi_2, c_1, c_2)\}$, we have $c\succ^*c'$ by the definition of the recursive path ordering on $T(F)$;
		\item If $\Gamma\vdash\phi =$ $AX_x(\phi_1)(s)$, and given that $\{\Gamma\vdash AX_x(\phi_1)(s), c_1, c_2\} \ssucc^* \{\vdash(s_1/x)\phi_1, \textsf{cpt}(\vdash (s_2/x)\phi_1, \textsf{cpt}(...\textsf{cpt}(\vdash(s_m/x)\phi_1, c_1, c_2), ..., c_2), c_2)\}  $ where $\textsf{Next}(s)=\{s_1,...,s_m\}$, we have $c\succ^*c'$ by the definition of the recursive path ordering on $T(F)$;
		\item The $EX$ case is analogous to the $AX$ case;
		\item If $\Gamma\vdash\phi = \Gamma\vdash EG_x(\phi_1)(s)$, 
		\begin{itemize}
			\item when $EG_x(\phi_1)(s)\in \Gamma$, then analogous to the first case, $c\succ^*c'$;
			\item when $EG_x(\phi_1)(s)\not\in \Gamma$,  $\textsf{Next}(s)=\{s_1,...,s_m\}$, and given that $\{\Gamma\vdash EG_x(\phi_1)(s), c_1, c_2\} \ssucc^*\{\vdash(s/x)\phi_1, \textsf{cpt}(\Gamma'\vdash EG_x(\phi_1)(s_1), c_1, \textsf{cpt}(...\textsf{cpt}(\Gamma'\vdash EG_x(\phi_1)(s_m), c_1, c_2)...), c_2)\}  $, we have $c\succ^*c'$ by the definition of the recursive path ordering on $T(F)$;
		\end{itemize}
		\item The cases of $AF, AR$, and $EU$ are analogous to the $EG$ case.
	\end{itemize}
\end{proof}

The correctness of this algorithm is ensured by the proposition below.

\begin{proposition}[Correctness of the Proof Search Algorithm]

Given a formula $\phi$, 
$\textsf{cpt}(\vdash\phi, \mathfrak{t}, \mathfrak{f}) \rsa^* \mathfrak{t}$ iff $\vdash \phi$ is provable,
	\label{coro:correct}
\end{proposition}

\begin{proof}
	Induction on the structure of $\phi$. The details are presented in 
\ref{appendix:proof:corret:prfsearch}.
\end{proof}

Moreover, for a given \textsf{CPT}
$\textsf{cpt}(\vdash \phi, \mathfrak{t}, \mathfrak{f})$, when
we memorize the already visited states to avoid visiting them again
during the proof of each modality, each sub-formula of $\phi$ appears,
in the worst case, $|\cal {\cal M}|$ times in the root of all
\textsf{CPT}s in the rewriting steps, where $|\cal {\cal M}|$ is the
number of states in the Kripke model under consideration. Therefore,
the time complexity of our proof search algorithm is
$O(|\phi| \times |{\cal M}|)$, where $|\phi|$ is the size of the
formula $\phi$ to be proved and $|{\cal M}|$ that of the model.

The pseudo code of the proof search algorithm is depicted in Figure~\ref{algotm:verification}.

In addition, the rewriting steps are memorized to 
build a proof tree after the end of the proof search.

\begin{figure}[!h]
	\centering
	\scriptsize
	\noindent\framebox{\parbox{.98\textwidth}{\hspace*{-0.3cm}

			\textsf{Input:} An input file $f$\\
			\textsf{Output:} A boolean result $r$\\
			\textsf{Name: }main
			\begin{algorithmic}
				\STATE Parse the input file $f$, and obtain the Kripke model $\cal M$, and a formula $\phi$ in the system $SCTL({\cal M})$;
				\STATE let $c = \mathsf{cpt}(\vdash\phi,\mathfrak{t}, \mathfrak{f})$;
				\WHILE{c is of the form $cpt(\Gamma\vdash\psi, c_1, c_2)$}
				\STATE rewrite $c$ to $c'$ in one step;
				\STATE $c\leftarrow c'$;
				\ENDWHILE
				\RETURN $c$;

			\end{algorithmic}
	}}
	
	\caption{The main algorithm}
	\label{algotm:verification}
\end{figure}

\begin{example}
	The use of rules in Figure \ref{fig:rewriting} are illustrated in Figure \ref{fig:cpt:rewriting},
	on the proof of Example \ref{exp2}.
	
\begin{figure}
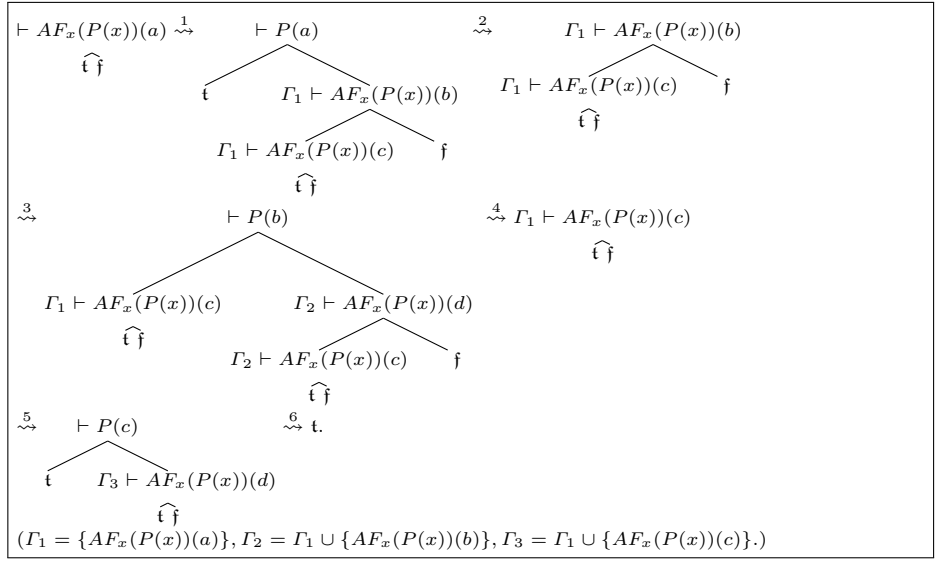

\scriptsize
\noindent\framebox{\parbox{.98\textwidth}{\hspace*{-0.1cm}

$
\Tree [.{$\vdash AF_x(P(x))(a)$} [.{$\mathfrak{t}$} ] [.{$\mathfrak{f}$} ] ]  
 \stackrel{1}\rsa
\Tree [.{$\vdash P(a)$} [.{$\mathfrak{t}$} ] [.{$\Gamma_1\vdash AF_x(P(x))(b)$} [.{$\Gamma_1\vdash AF_x(P(x))(c)$} [.{$\mathfrak{t}$} ] [.{$\mathfrak{f}$} ] ] [.{$\quad \quad \mathfrak{f}\quad$} ] ] ]    
\stackrel{2}\rsa
\Tree [.{$\Gamma_1\vdash AF_x(P(x))(b)$} [.{$\Gamma_1\vdash AF_x(P(x))(c)$} [.{$\mathfrak{t}$} ] [.{$\mathfrak{f}$} ] ] [.{$\quad \quad \mathfrak{f}\quad$} ] ]\\ 
\stackrel{3}\rsa
\Tree [.{$\vdash P(b)$} [.{$\Gamma_1\vdash AF_x(P(x))(c)$} [.{$\mathfrak{t}$} ] [.{$\mathfrak{f}$} ] ] [.{$\Gamma_2\vdash AF_x(P(x))(d)$} [.{$\Gamma_2\vdash AF_x(P(x))(c)$} [.{$\mathfrak{t}$} ] [.{$\mathfrak{f}$} ] ] [.{$\quad\quad\mathfrak{f}\quad$} ]  ] ]   \stackrel{4}\rsa
\Tree [.{$\Gamma_1\vdash AF_x(P(x))(c)$} [.{$\mathfrak{t}$} ] [.{$ \mathfrak{f}$} ] ] \\
\stackrel{5} \rsa
\Tree [.{$\vdash P(c)$} [.{$\mathfrak{t}$} ] [.{$\quad \quad \Gamma_3\vdash AF_x(P(x))(d)$} [.{$\mathfrak{t}$} ] [.{$\mathfrak{f}$} ] ] ] 
\stackrel{6}\rsa
\mathfrak{t}. \\
(\Gamma_1 = \{AF_x(P(x))(a)\}, \Gamma_2 = \Gamma_1 \cup \{AF_x(P(x))(b)\}, \Gamma_3 = \Gamma_1 \cup \{AF_x(P(x))(c)\}.)
$
 
}}
\caption{An illustration of \textsf{CPT} rewritings.}
\label{fig:cpt:rewriting}
\end{figure}

\hspace{-0.5cm}
{\it Step 1}. 
At this step, on the left side of $\stackrel{1}\rsa$, the root of the CPT is $\vdash AF_x(P(x))(a)$.
We need to decide whether $\vdash AF_x(P(x))(a)$ is provable, which is not known at that moment yet. 
So we have to decide first whether $P(a)$ is provable, 
and then both $AF_x(P(x))(a)\vdash AF_x(P(x))(b)$ and $AF_x(P(x))(a)\vdash AF_x(P(x))(c)$ are successively provable, 
corresponding  applying the $\mathbf{AF}$-{$\mathsf{R_1}$} rule and the $\mathbf{AF}$-{$\mathsf{R_2}$} rule, respectively. 
We encode those two steps in a single \textsf{CPT}, which is the one on the right side of $\stackrel{1}\rsa$.

\hspace{-0.5cm}
{\it Step 2}. 
Since the atomic formula $P(a)$ is not provable, 
the \textsf{CPT} on the left side of $\stackrel{2}\rsa$ reduces to its right subtree (f-continuation), 
which is the \textsf{CPT} on the right side of $\stackrel{2}\rsa$.
 
\hspace{-0.5cm}
{\it Step 3}. 
Like at step 1, we need to decide whether $AF_x(P(x))(a)\vdash AF_x(P(x))(b)$ is provable, which is not known at that moment yet.
So we encode the left subtree (t-continuation) of the \textsf{CPT} which is on the left side of $\stackrel{3}\rsa$, 
and, by the $\mathbf{AF}$-{$\mathsf{R_1}$} rule and the $\mathbf{AF}$-{$\mathsf{R_2}$} rule,  
the two steps to find successively the proofs of $\vdash P(b)$ and of $AF_x(P(x))(a), AF_x(P(x))(b)\vdash AF_x(P(x))(d)$ into the \textsf{CPT} 
which is on the right side of $\stackrel{3}\rsa$.

\hspace{-0.5cm}
{\it Step 4}. 
Like at step 2, we can judge the atomic formula $P(b)$ is provable immediately.
So the \textsf{CPT} on the left side of $\stackrel{4}\rsa$ reduces to its left subtree (t-continuation) 
which is on the right side of $\stackrel{4}\rsa$.

\hspace{-0.5cm}
{\it Step 5}. 
Like at step 1 and 3, we can not judge whether the sequent $AF_x(P(x))(a)\vdash AF_x(P(x))(c)$ is provable immediately, 
so we encode the two steps to find successively the proofs of $\vdash P(c)$ and $AF_x(P(x))(a), AF_x(P(x))(b)\vdash AF_x(P(x))(d)$
into the \textsf{CPT} which is on the right side of $\stackrel{5}\rsa$;

\hspace{-0.5cm}
{\it Step 6}. 
Like at step 2 and 4, as the atomic formula $P(c)$ is provable, 
so the \textsf{CPT} on the left side of $\stackrel{6}\rsa$ reduces to its left subtree (t-continuation) which is $\mathfrak{t}$,
Now, the proof search of $\vdash AF_x(P(x))(a)$ terminates, and we can judge that this sequent is provable.
\end{example}

\subsubsection{Memorization}

In the proof search of sequents with co-inductive formulae (formulae with modality $EG$ or
$AR$), the merge rules are used to assert that some property holds on
an infinite path of states. For every merge rule, the formulae need to be
memorized are with the same modality, whereas the only differences are
the states appearing in the formulae. Thus, it is sufficient to
memorize only the states, not the whole formulae, in the
implementation of every merge (i.e., $\Gamma$). Essentially, each
construction of a merge is implemented by memorizing an infinite path
where all states verify some property. 

What is worth mentioning is the proof search of sequents with
inductive formulae (formulae with modality $AF$ or $EU$). Although
there are no merge rules for the proof of this kind of sequents,
merges are also helpful to avoid infinite proof search, when the
formula is not provable, that is when its (co-inductive) negation
is. For instance, for the proof search of the sequent
$\vdash EU_{x,y}(\phi_1,\phi_2)(s)$, we need to find a finite path
where in the last states $\phi_2$ holds, and in all other states
$\phi_1$ holds. Although we are not finding infinite paths, we still
need to avoid our proof search falling into an infinite path. Thus, as
an optimization in the proof search of inductive formulae, we also
keep merges in the rewriting rules.  Note that merges are not
reflected in the proof rules for the $AF$ and $EU$ cases. The reason
is that, in the proof rules, we only care about the shape of the proof
tree, not how the proof tree is constructed. It is only in the
construction of proof trees where merges for $AF$ and $EU$ are
mentioned.

As another optimization of the proof search algorithm, we use a global memory to remember, for each sub-formula $\phi$, 
the states visited during the proof search of this formula, and avoid visiting states that are already in this memory. 
This memory can either be a hash table or a \BDD{}, each having 
advantages and disadvantages. 
This memory helps to avoid constructing the same merges
repeatedly. This optimization does not break the correctness
property of the proof search algorithm, as we are only omitting
repeatedly rewriting steps on \textsf{CPT}s.

\subsection{Relations with some model checking techniques}
In this section, we discuss the relations of the techniques adopted in \sctl{} with those in some other \CTL{} model checking approaches.

\subsubsection{\BDD{}-based symbolic model checking}
When a Kripke model contains mostly boolean variables, for instance in model checking for hardware problems, 
using \BDD{}s to memorize states is an effective way to reduce space during verification procedure.	
The best known \BDD{}-based symbolic model checker is \nusmv{} \cite{mcmillan93,CimattiCGR99}, 
and its extension \nuxmv{} \cite{CAVCDGMMMRT14}.
To illustrate verification procedure in a \BDD{}-based symbolic model checker, 
let us consider, for instance, a Kripke model with the initial state $s_0$ and a transition relation $T$.  
To check whether ${\cal M},s_0\models EF\phi$ holds in such a model checker, say \nusmv, 
first one needs to calculate the least fixed point
$\textup{lfp}=\mu Y. (\phi\vee EXY)$, 
then check whether $s_0\in \textup{lfp}$ \cite{mcmillan93,CimattiCGR99}. 
Calculating the $\textup{lfp}$ corresponds to unfolding the transition relation $T$,
where states that are not reachable from $s_0$ may be involved. 
	
The verification procedure in \sctl{} differs from traditional \CTL{} symbolic model checkers.
For instance, unlike in \nusmv, there is no need for \sctl{} to calculate a fixed point of the transition relation. 
Instead, unfolding of the transition relation stops as soon as the given property is proved or its negation is proved. 
Moreover, 
\sctl{} can memorize visited states either directly when there are many non-boolean variables in the model,
or using \BDD{}s when the model contains mostly boolean variables.
In the latter case, unlike \nusmv{} that encodes models and properties into \BDD{}s before searching state space, 
\sctl{} searches states directly on the Kripke model under consideration, 
using \BDD{}s to memorize the visited states only. 

\subsubsection{On-the-fly model checking}
The on-the-fly style of searching state space helps avoid exploring unneeded states. 	
Indeed, in on-the-fly model checking, 
usually, there is no need to generate the full state space.
Traditional  on-the-fly \CTL{} model checking algorithms \cite{VergauwenL93,BCG95} are usually recursive, i.e., the unfolding of the formula and the transition rules are preformed recursively. 
These recursive based algorithms usually involves a lot of stack operations when verifying properties over big size Kripke models. These stack operations may consume much time during the verification processes.

In \sctl{}, the proof search of a formula mimics a double on-the-fly style model checking, that is,
unfolding on demand both transition relations and the formulae.
However, unlike traditional on-the-fly model checking algorithms, our algorithm is in continuation-passing style, which contains only constant stack operations \cite{Sestoft12}. In the programming language theory, 
a {\em continuation} is an explicit representation of the {\em the rest of the computation}.  
A function is said in continuation-passing style (\textsf{CPS}), 
if it takes an extra argument, the continuation, 
which decides what will happen to the result of the function.  
This method, usually used in compiling and programming, can help, among others, 
to reduce considerably the size of the stacks \cite{Reynolds93,Appel06,Sestoft12}.
	
We would like to compare our algorithm in \sctl{}
to those given in \cite{VergauwenL93} and \cite{BCG95}, respectively.
However, as far as we know, there are no tools based on these algorithms
that can fully solve \CTL{} model checking problems. 
To show that using continuation-passing style is not a trivial improvement, 
we designed therefore a recursion variant of \sctl{}, 
called \sctlprovr{}\footnote{\url{https://github.com/terminatorlxj/SCTLProV_R}}.  
The difference between \sctlprovr{} and \sctl{} is that, instead of using continuations, 
\sctlprovr{} uses recursion calls to prove sub-formulae and search state space. 
We will compare the experimental results of \sctl{} and \sctlprovr{} in Section \ref{subsec:random}.
	
\subsubsection{Bounded model checking}
For traditional \textsf{BMC} tools, 
where the temporal formulae under proving are unfolded on a set of
traces with limited length once for all.
For example, in model checking ${\cal M}, s_0 \models_{k+1} EF\phi$,
one unfolding step of the $EF$ formula involves $k+1$ unfolding steps of the transition relation $T$, 
that is, \BMC{} tools need to deal with the bulky formula \cite{BCCZ99}:
	
	\begin{scriptsize}
		$$ [{\cal M},EF\phi]_{k+1} := \bigwedge^{k-1}_{i=0}T(s_i,s_{i+1}) \wedge \bigvee_{j=0}^k\phi(s_j)$$
	\end{scriptsize}
	
To avoid exploring unnecessary states in $M$,
\sctl{} unfolds on demand the transition relation $T$.	
Thus, in \SCTL{},  one unfolding step of a formula involves at most one unfolding step of the transition relation.
In fact, to verify $\vdash EF_x(\phi)(s_0)$, 
\sctl{} unfolds the transition relation $T$ and the formula $EF_x(\phi)(s_0)$ as
	\begin{center}{\scriptsize
			$\begin{array}{l}
			\textsf{unfold}(S,EF_x(\phi)(s_i)) := \\
			\phi(s_i) \vee ((s_i\notin S) \wedge T(s_i, s_{i+1}) \wedge 
			\textsf{unfold}(S\cup \{s_i\},EF_x(\phi)(s_{i+1})))
			\end{array}$
		}
	\end{center}
	where $S$ is a set representing the visited states during the proof search, 
	which is in fact our implementation of the \textsf{merge} rule of Figure \ref{sctl_rules}.

\section{\SCTL{} with fairness constraints}\label{sect:fairness}

Fairness is an important aspect in verifying concurrent
systems. Fairness assumptions often rule out unrealistic behaviors,
and are often necessary to establish liveness properties
\cite{BaierKatoen08}. For instance, in a mutual exclusion algorithm of
two processes, we usually need to consider a fair scheduling of the
execution of the processes, i.e., no process waits infinitely
long. Such fairness constraints can also be defined in \SCTL{}. Our
definition of fairness coincides with that in \cite{mcmillan93}, i.e.,
the path quantifiers apply to those paths along which each formula in
a set $C$ holds infinitely often. For instance, $E_Cf$ means that
there exists a path such that each formula in $C$ is true infinitely
often and $f$ is true in this path.

We define the fairness constraint $C$ as a set of \SCTL{} formulae; an
infinite path is fair under fairness constraint $C$ if and only if for
each \SCTL{} formula $\phi\in C$, $\phi$ is valid infinitely often on
this path.  Formula $E_CG_x(\phi)(t)$ is valid if and only if there
exists an infinite path, fair under $C$, starting from state $t$ such
that for all state $s$ in this path, $(s/x)\phi$ is valid. Similarly,
formula $A_CF_x(\phi)(t)$ is valid if and only if for each infinite
path, fair under $C$, starting from state $t$ such that there exists a
state $s$ on this path and $(s/x)\phi$ is valid.

Similar to \cite{mcmillan93}, other \SCTL{} formulae with fairness constraints can be characterized in terms of $E_CG$ formulae and $A_CF$ formulae:
\begin{small}
	$$E_CX_x(\phi)(t) = EX_x(\phi \wedge E_CG_x(\top)(x))(t)$$
	$$A_CX_x(\phi)(t) = AX_x(\phi \vee A_CF_x(\bot)(x))(t)$$
	$$E_CU_{x,y}(\phi_1,\phi_2)(t) = EU_{x,y}(\phi_1, \phi_2\wedge E_CG_z(\top)(y))(t)$$
	$$A_CR_{x,y}(\phi_1,\phi_2)(t) = AR_{x,y}(\phi_1, \phi_2\vee A_CF_z(\bot)(y))(t)$$
\end{small}

Given that \SCTL{} is sound and complete, to prove $E_CG_x(\phi)(t)$ is equivalent to prove $EG_x(\phi)(t)$ where only fair paths are considered, i.e., to prove the existence of a fair path on which $\phi$ is always provable. Similarly, to prove $A_CF_x(\phi)(t)$ is equivalent to prove $AF_x(\phi)(t)$ where only fair paths are considered, i.e., to prove the absence of a fair path on which $\phi$ is always not provable. Thus, to prove \SCTL{} formulae with fairness constraints, we need a mechanism to decide the existence of fair paths.

According to Proposition~\ref{prop:fair_if} and Proposition~\ref{prop:fair_fi} shown below, we can decide the existence of a fair path in finite steps, which is exactly the purpose of our merges. To be more precise, when the merge rule is applied, we check the fairness of the path constructed and discard those that are not fair: i.e., we only consider merges where each formula in $C$ is provable in some state of a loop.

\begin{proposition}\label{prop:fair_if}
	For a set $C$ of \SCTL{} formulae and an infinite sequence of states $\sigma = s_0,s_1,...$ such that for all $i, s_i \rightarrow s_{i+1}$, if each element of $C$ is valid infinitely often in $\sigma$, then there exists a finite sequence of states $\sigma_f = s'_0,s'_1,...,s'_n$ such that for all $0\le j\le n-1$, $s'_j\rightarrow s'_{j+1}$, and there exists $0\le p\le n-1$ such that $s'_n = s'_p$, all the $s'_j$ are among $\sigma$, and for each element $f \in C$, $f$ is valid in some state $s'_q$, where $p\le q\le n$ . 
\end{proposition}
\begin{proof}
	As the number of states is finite, there exists a set of states $S$, such that each state $s\in S$ appears infinitely often in $\sigma$, and each formula $f\in C$ is valid in some state in $S$. Otherwise, if for each set $S'$ of states that occur infinitely often in $\sigma$, there exists some formula $f\in C$ such that $f$ is not valid in any element of $S'$, then $f$ is not valid in any state of $\sigma$ that occur infinitely often, and thus $f$ is not valid infinitely often. Assume $S = \{s_{i_1},s_{i_2},...,s_{i_k}\}$ such that $i_1\le i_2\le ...\le i_k$, then let $s'_p = s_{i_1}$, and $s'_n = s_{i_{k'}}$ such that $i_{k'} \ge i_k$ and $s_{i_{k'}} = s_{i_1}$.
\end{proof}

\begin{proposition}\label{prop:fair_fi}
	For a set $C$ of \SCTL{} formulae and a finite sequence of states $\sigma_f = s_0,s_1,...,s_n$ such that for all $0\le i\le n-1$, $s_i\rightarrow s_{i+1}$, there exists $0\le p\le n$, $s_p = s_n$, and every formula in $C$ is valid in some states between $s_p$ and $s_n$ in the sequence, then there exists an infinite sequence of states $\sigma = s'_0, s'_1,...$ such that for all $i$, $s'_i \rightarrow s'_{i+1}$, all the $s'_j$ are among $s_0,s_1,...,s_n$, and every formula in $C$ is valid infinitely often in the infinite sequence.
\end{proposition}
\begin{proof}
The sequence $\sigma = s_0,...,s_{p-1},s_p,...,s_{n-1},...$ 
verifies the properties above.
\end{proof}

\section{Example and Experimental Evaluation}\label{experiment}
To illustrate the feasibility and the efficiency of \sctl{}, we first use an example (Subsection~\ref{subsec:example}) to show an application of \sctl{}, and then evaluate several benchmarks (benchmark \#1, \#2 and \#3 in Subsection~\ref{subsec:random}, and benchmark \#4 in Subsection~\ref{subsec:fair}) to show the efficiency of \sctl{}, and compare the experimental results with four other verification tools: the Resolution-based
theorem prover \tool{iProver Modulo} \cite{Burel11}, the \QBF{}-based
bounded model checker \verds{} version 1.49, the
\BDD-based unbounded model checker \nusmv{} version 2.6.0 and
its extension \nuxmv{} version 1.0.0. 
All examples and benchmarks are tested on a Linux
platform with 3.0 GB memory and a 2.93GHz $\times$ 4 CPU, and the time limit is 20 minutes.

All benchmarks used in this paper are available online\footnote{\url{https://github.com/terminatorlxj/ctl_benchmarks}}.

\subsection{An illustrative example}\label{subsec:example}
\begin{example}[A Mutual Exclusion Problem]

  This example is a mutual exclusion algorithm of two concurrent
  processes (process $A$ and process $B$) described in
  \cite{Peterson81}. \textit{Mutual Exclusion} means that both two
  processes can not enter the critical section at the same time. This
  problem is addressed in several model
  checkers. In our formulation of this problem, a shared variable
  $mutex$ is used to remember the number of processes that have
  entered the critical section. A violation of \textit{Mutual
    Exclusion} means that in some state of the program, the value of
  the shared variable $mutex$ is $2$.
\begin{figure}[h]
\centering
\scriptsize
\begin{verbatim}
Model mutual()
{
  Var {
     flag : Bool; mutex : (0 .. 2); 
     a : (1 .. 6); b : (1 .. 6);
  }
  Init {
     flag := false; mutex := 0; a := 1; b := 1;
  }
  Transition {
     a = 1 && flag = false : {a := 2;};
     a = 2 : {a := 3; flag := true;};
     /*A has entered the critical section*/
     a = 3 : {a := 4; mutex := mutex + 1;};
     /*A has left the critical section*/ 
     a = 4 : {a := 5; mutex := mutex - 1;};
     a = 5 : {a := 6;};
     b = 1 && flag = false : {b := 2;};
     b = 2 : {b := 3; flag := true;};
     /*B has entered the critical section*/
     b = 3 : {b := 4; mutex := mutex + 1;};
     /*B has left the critical section*/ 
     b = 4 : {b := 5; mutex := mutex - 1;};
     b = 5 : {b := 6;};
  }
  Atomic {bug(s) := s(mutex) = 2;}
  Spec{find_bug := EU(x, y, TRUE, bug(y), ini);}
}
\end{verbatim}
\caption{The input file ``mutual.model".}
\label{fig:mutual}
\end{figure}

In the input file (Figure \ref{fig:mutual}), variable $flag$ is a signal indicating whether there exists a process is running; Variables $a$ and $b$ indicate the program counters of the two processes, respectively. The property to be checked is that whether both processes are in the critical section at the the same time. We check this property in \sctl{} using the following command:
\begin{center}
\scriptsize
\begin{verbatim}
sctl -output output.out mutual.model
\end{verbatim}
\end{center}

The result is as follows, which indicates that there is a bug in the mutual exclusion problem, i.e., the mutual exclusion property is violated.
\begin{center}
\scriptsize
\begin{verbatim}
verifying on the model mutual...
find_bug: EU(x,y, TRUE, bug(y), ini)
find_bug is true.
\end{verbatim}
\end{center}

The proof tree of the property is output to the file ``output.out".
\begin{center}
\scriptsize
\begin{verbatim}
0: |- EU(x,y,TRUE,bug(y),{flag:=false;mutex:=0;a:=1;b:=1})	[4, 1]
4: {flag:=false;mutex:=0;a:=1;b:=1}
|- EU(x,y,TRUE,bug(y),{flag:=false;mutex:=0;a:=2;b:=1})	[7, 5]
1: |- TRUE	[]
7: {flag:=false;mutex:=0;a:=1;b:=1} 
   {flag:=false;mutex:=0;a:=2;b:=1}
|- EU(x,y,TRUE,bug(y),{flag:=false;mutex:=0;a:=2;b:=2})	[23, 20]
5: |- TRUE	[]
23:{flag:=false;mutex:=0;a:=1;b:=1} 
   {flag:=false;mutex:=0;a:=2;b:=1} 
   {flag:=false;mutex:=0;a:=2;b:=2}
|- EU(x,y,TRUE,bug(y),{flag:=true;mutex:=0;a:=3;b:=2})	[27, 24]
20: |- TRUE	[]
27:{flag:=false;mutex:=0;a:=1;b:=1} 
   {flag:=false;mutex:=0;a:=2;b:=1} 
   {flag:=false;mutex:=0;a:=2;b:=2} 
   {flag:=true;mutex:=0;a:=3;b:=2}
|- EU(x,y,TRUE,bug(y),{flag:=true;mutex:=1;a:=4;b:=2})	[31, 28]
24: |- TRUE	[]
31:{flag:=false;mutex:=0;a:=1;b:=1} 
   {flag:=false;mutex:=0;a:=2;b:=1} 
   {flag:=false;mutex:=0;a:=2;b:=2} 
   {flag:=true;mutex:=0;a:=3;b:=2} 
   {flag:=true;mutex:=1;a:=4;b:=2}
|- EU(x,y,TRUE,bug(y),{flag:=true;mutex:=1;a:=4;b:=3})	[35, 32]
28: |- TRUE	[]
35:{flag:=false;mutex:=0;a:=1;b:=1} 
   {flag:=false;mutex:=0;a:=2;b:=1} 
   {flag:=false;mutex:=0;a:=2;b:=2} 
   {flag:=true;mutex:=0;a:=3;b:=2} 
   {flag:=true;mutex:=1;a:=4;b:=2} 
   {flag:=true;mutex:=1;a:=4;b:=3}
|- EU(x,y,TRUE,bug(y),{flag:=true;mutex:=2;a:=4;b:=4})	[37]
32: |- TRUE	[]
37: |- bug({flag:=true;mutex:=2;a:=4;b:=4})	[]
\end{verbatim}
\end{center}
According to the output above, we can find that after process $A$ have entered the critical section, process $B$ can also enter the critical section.

\end{example}

\subsection{Randomly generated programs}\label{subsec:random}
We consider three benchmarks in this part. 
The original description of benchmark \#1 is in \cite{Zhang14} and also restated here.
Based on benchmark \#1, we extend the number of variables to tens, hundreds, and even thousands in benchmark \#2 and benchmark \#3.
The randomness of the test cases in three benchmarks makes it rather fair for different \CTL{} model checking approaches, and helps us recognize the strengths and weaknesses of each tool. 

\subsubsection{Benchmark \#1}
Benchmark \#1 chosen in this subsection is originally introduced by Zhang \cite{Zhang14} in the evaluation of model checkers \verds{} and \nusmv{}. Later, Ji \cite{Ji15} also uses this benchmark in the evaluation of the theorem prover \tool{iProver Modulo} and the model checker \verds{}. This benchmark consists of 2880 randomly generated test cases where two types of random Boolean programs are considered---Concurrent Processes and Concurrent Sequential Processes. 
In programs with Concurrent Processes,
the parameters of the first set of random Boolean programs are as
follows.
\begin{center}
	\begin{tabular}{|l|}
		\hline
		$a$: number of processes \\
		$b$: number of all variables \\
		$c$: number of shared variables \\
		$d$: number of local variables in a process \\
		\hline
	\end{tabular}
\end{center}
The shared variables are initially set to a random value in $\{0,1\}$,
and the local variables are initially set to $0$. For each process,
the shared variables and the local variables are assigned the negation
of a variable randomly chosen from these variables. We test different
sizes of the programs with 3 processes ($a=3$), and let $b$ vary over
the set of values $\{12,24,36\}$, then set $c=b/2, d=c/a$. Each of the
24 properties is tested on 20 test cases for each value of $b$.

In programs with Concurrent Sequential Processes,
in addition to $a,b,c,d$ specified above, the parameters of the second set of random Boolean programs are as
follows.
\begin{center}
	\begin{tabular}{|l|}
		\hline
		$t$: number of transitions in a process \\
		$p$: number of parallel assignments in each transition \\
		\hline
	\end{tabular}
\end{center}
For each concurrent sequential process, besides the $b$ Boolean
variables, there is a local variable representing program locations,
with $c$ possible values. The shared variables are initially set to a
random value in $\{0,1\}$, and the local variables are initially set
to $0$. For each transition of a process, $p$ pairs of shared
variables and local variables are randomly chosen among the shared
variables and the local variables, such that the first element of such
a pair is assigned the negation of the second element of the
pair. Transitions are numbered from $0$ to $t-1$, and are executed
consecutively, and when the end of the sequence of the transitions is
reached, it loops back to the execution of the transition numbered
$0$. For this type of programs, we test different sizes of the
programs with $2$ processes ($a=2$), and let $b$ vary in the set of
values $\{12,16,20\}$, and then set $c=b/2, d=c/a, t=c$, and
$p=4$. Similarly, each property is tested on $20$ test cases for each
value of $b$.

Twenty-four properties are to be checked in this benchmark: properties $P_{01}$ to $P_{12}$ are depicted in Figure~\ref{fig:properties}, and $P_{13}$ to $P_{24}$ are simply the variations of
$P_{01}$ to $P_{12}$ by replacing $\wedge$ and $\bigvee$ by $\vee$ and
$\bigwedge$, respectively.

\begin{figure}[!h]
	\centering
	{\scriptsize
		\begin{tabular}{|l|l|}
			\hline
			$P_{01}$& $AG(\bigvee^c_{i=1}v_i)$ \\
			\hline
			$P_{02}$& $AF(\bigvee^c_{i=1}v_i)$ \\
			\hline
			$P_{03}$& $AG(v_1 \A AF(v_2\wedge \bigvee^c_{i=3}v_i))$
			\\
			\hline
			$P_{04}$& $AG(v_1 \A EF(v_2\wedge \bigvee^c_{i=3}v_i))$
			\\
			\hline
			$P_{05}$& $EG(v_1 \A AF(v_2\wedge \bigvee^c_{i=3}v_i))$
			\\
			\hline
			$P_{06}$& $EG(v_1 \A EF(v_2\wedge \bigvee^c_{i=3}v_i))$
			\\
			\hline
			$P_{07}$& $AU(v_1, AU(v_2, \bigvee^c_{i=3}v_i))$\\
			\hline
			$P_{08}$ & $AU(v_1, EU(v_2, \bigvee^c_{i=3}v_i))$\\
			\hline
			$P_{09}$& $AU(v_1, AR(v_2, \bigvee^c_{i=3}v_i))$\\
			\hline
			$P_{10}$& $AU(v_1, ER(v_2, \bigvee^c_{i=3}v_i))$\\
			\hline
			$P_{11}$& $AR(AX v_1, AX AU(v_2, \bigvee^c_{i=3}v_i))$\\
			\hline
			$P_{12}$& $AR(EX v_1, EX EU(v_2, \bigvee^c_{i=3}v_i))$\\
			\hline
		\end{tabular}
	}
	\caption{Properties $P_{01}, P_{02}, \ldots, P_{12}$ to be checked in benchmark \#1, \#2, and \#3.}
	\label{fig:properties}
\end{figure}

\subsubsection{Benchmark \#2 and \#3}
In benchmark {\#}2, we increase the number of state variables in benchmark \#1 to $48$, $60$, or $72$ for Concurrent Processes, and $24$,
$28$, or $32$ for Concurrent Sequential Processes. The 2880 test cases are also randomly generated.
The properties to be checked are the same as in benchmark \#1.

In benchmark {\#}3, we increase the number of state variables in benchmark \#1 to $252$, $504$
and $1008$ for both Concurrent Processes and Concurrent Sequential Processes, and check the same properties as benchmark \#1 and \#2.

\subsubsection{Experimental data}
The experimental results are shown below, and the detailed data is in \ref{app:detail:data}.

\paragraph{Experimental data for benchmark \#1.}
For 2880 test cases in this benchmark, \tool{iProver Modulo} can solve 1816 (63.1\%) cases, \verds{} can solve 2230 (77.4\%) cases, \sctl{} can solve 2862 (99.4\%) cases, and both \nusmv{} and \nuxmv{} can solve all (100\%) test cases. The numbers of test cases where \sctl{} runs faster are 2823 (98.2\%) comparing with \tool{iProver Modulo}, 2858 (99.2\%) comparing with \verds{}, 2741 (95.2\%) comparing with \nusmv{}, and 2763 (95.9\%) comparing with \nuxmv{}. According to Figure~\ref{fig:average_time} and Figure~\ref{fig:average_memory}, \sctl{} uses less time and space than the other four tools.
\paragraph{Experimental data for benchmark \#2.}
For 2880 test cases in this benchmark, \tool{iProver Modulo} can solve 1602 (55.6\%) cases, \verds{} can solve 1874 (65.1\%) cases, \nusmv{} can solve 728 (25.3\%) cases, \nuxmv{} can solve 736 (25.6\%) cases, and \sctl{} can solve 2597 (90.2\%) cases. The numbers of test cases where \sctl{} runs faster are 2597 (90.2\%) comparing with \tool{iProver Modulo}, 2594 (90.1\%) comparing with \verds{}, and 2588 (89.9\%) comparing both with \nusmv{} and \nuxmv{}. According to Figure~\ref{fig:average_time:extended} and Figure~\ref{fig:average_memory:extended}, \sctl{} uses less time and space than the other four tools.
\paragraph{Experimental data for benchmark \#3.}
For 2880 test cases in this benchmark, \tool{iProver Modulo} can solve 1146 (39.8\%) cases, \verds{} can solve 352 (12.2\%) cases, \sctl{} can solve 1844 (64.0\%) cases, while neither \nusmv{} nor \nuxmv{} can solve any case.

\begin{figure}[h]\tiny\centering
	\begin{tabular}{c}
		\begin{tikzpicture}[scale=0.55]
		\begin{axis}[title={\small{CP}},legend pos=north west, small,
		xlabel = {Number of state variables},
		ylabel = {Time [seconds]}
		]
		\addplot [color=red, mark=x] coordinates
		{
			(12,0.011)
			(24,0.010)
			(36,0.057)     
		};
		\addplot [color=green, mark=triangle] coordinates
		{
			(12,6.293)
			(24,14.648)
			(36,16.351)
		};
		\addplot [color=blue, mark=o] coordinates
		{
			(12,1.904)
			(24,0.714)
			(36,19.200)
		};
		\addplot [color=black,mark=*] coordinates
		{
			(12,0.014)
			(24,2.202)
			(36,135.202)
		};
		\addplot [color=black,mark=square] coordinates
		{
			(12,0.018)
			(24,2.100)
			(36,130.268)
		};
		\legend{\sctl{},\tool{iProver Modulo}, \verds{}, \nusmv{}, \nuxmv{}}
		\end{axis}
		\end{tikzpicture}
		
	\end{tabular}
	\begin{tabular}{c}
		\scriptsize
		\begin{tikzpicture}[scale=0.55]
		\begin{axis}[title={\small{CSP}},legend pos=north west, small,
		xlabel = {Number of state variables},
		ylabel = {Time [seconds]}
		]
		\addplot [color=red, mark=x] coordinates
		{
			(12,0.006)
			(16,0.007)
			(20,0.374)      
		};
		\addplot [color=green, mark=triangle] coordinates
		{
			(12,4.995)
			(16,3.997)
			(20,4.424)
		};
		\addplot [color=blue, mark=o] coordinates
		{
			(12,20.203)
			(16,75.741)
			(20,136.387)
		};
		\addplot [color=black,mark=*] coordinates
		{
			(12,0.105)
			(16,2.036)
			(20,53.195)
		};
		\addplot [color=black,mark=square] coordinates
		{
			(12,0.107)
			(16,1.957)
			(20,49.144)
		};
		\tiny{\legend{\sctl{},\tool{iProver Modulo}, \verds{}, \nusmv{}, \nuxmv{}}}
		\end{axis}
		\end{tikzpicture}
	\end{tabular}
	
	\caption{Average verification time in benchmark {\#}1.}
	\label{fig:average_time}
	
	\begin{subfigure}\centering
		\begin{tabular}{c}
			\begin{tikzpicture}[scale=0.55]
			\begin{axis}[title={\small{CP}},legend pos=north west, small, ymax=1000,
			xlabel = {Number of state variables},
			ylabel = {Memory [MB]}
			]
			\addplot [color=red, mark=x] coordinates
			{
				(12,7.845)
				(24,22.328)
				(36,42.184)     
			};
			\addplot [color=green, mark=triangle] coordinates
			{
				(12,10.111)
				(24,16.547)
				(36,73.946)
			};
			\addplot [color=blue, mark=o] coordinates
			{
				(12,322.020)
				(24,468.169)
				(36,581.011)
			};
			\addplot [color=black,mark=*] coordinates
			{
				(12,8.818)
				(24,42.924)
				(36,251.364)
			};
			\addplot [color=black,mark=square] coordinates
			{
				(12,21.013)
				(24,55.179)
				(36,256.058)
			};
			\legend{\sctl{},\tool{iProver Modulo}, \verds{}, \nusmv{}, \nuxmv{}}
			\end{axis}
			\end{tikzpicture}
			
		\end{tabular}
	\end{subfigure}
	\begin{subfigure}
		\centering
		\begin{tabular}{c}
			\scriptsize
			\begin{tikzpicture}[scale=0.55]
			\begin{axis}[title={\small{CSP}},legend pos=north west, small, ymax=1000,
			xlabel = {Number of state variables},
			ylabel = {Memory [MB]}
			]
			\addplot [color=red, mark=x] coordinates
			{
				(12,1.984)
				(16,2.039)
				(20,3.383)      
			};
			\addplot [color=green, mark=triangle] coordinates
			{
				(12,10.070)
				(16,11.449)
				(20,23.660)
			};
			\addplot [color=blue, mark=o] coordinates
			{
				(12,322.023)
				(16,485.081)
				(20,514.027)
			};
			\addplot [color=black,mark=*] coordinates
			{
				(12,7.051)
				(16,29.151)
				(20,144.974)
			};
			\addplot [color=black,mark=square] coordinates
			{
				(12,21.168)
				(16,48.423)
				(20,157.113)
			};
			\tiny{\legend{\sctl{},\tool{iProver Modulo}, \verds{}, \nusmv{}, \nuxmv{}}}
			\end{axis}
			\end{tikzpicture}
		\end{tabular}
	\end{subfigure}
	
	\caption{Average memory usage in benchmark {\#}1.}
	\label{fig:average_memory}
\end{figure}

\begin{figure}[h!]\tiny\centering
	\begin{subfigure}\centering
		\begin{tabular}{c}
			
			\begin{tikzpicture}[scale=0.55]
			\begin{axis}[title={\small{CP}},legend pos=north west, small,
			xlabel = {Number of state variables},
			ylabel = {Time [seconds]}
			]
			\addplot [color=red, mark=x] coordinates
			{
				(12,0.011)
				(24,0.280)
				(36,2.929)  
				(48,5.100)
				(60,7.357)   
			};
			\addplot [color=green, mark=triangle] coordinates
			{
				(12,6.293)
				(24,14.648)
				(36,16.351)
				(48,20.130)
				(60,37.303) 
			};
			\addplot [color=blue, mark=o] coordinates
			{
				(12,1.904)
				(24,0.714)
				(36,19.200)
				(48,40.825)
				(60,80.201)
			};
			\addplot [color=black,mark=*] coordinates
			{
				(12,0.014)
				(24,2.202)
				(36,135.202)
				(48,477.578)
				(60,1095.582)
			};
			\addplot [color=black,mark=square] coordinates
			{
				(12,0.018)
				(24,2.100)
				(36,130.268)
				(48,450.324)
				(60,995.689)
			};
			\legend{\sctl{}, \tool{iProver Modulo}, \verds{}, \nusmv{}, \nuxmv{}}
			\end{axis}
			\end{tikzpicture}
			
		\end{tabular}
	\end{subfigure}
	\begin{subfigure}
		\centering
		\begin{tabular}{c}
			\scriptsize
			\begin{tikzpicture}[scale=0.55]
			\begin{axis}[title={\small{CSP}},legend pos=north west, small,
			xlabel = {Number of state variables},
			ylabel = {Time [seconds]}
			]
			\addplot [color=red, mark=x] coordinates
			{
				(12,0.006)
				(16,0.007)
				(20,0.374) 
				(24,9.903)
				(28,12.548)
				(32,26.417)     
			};
			\addplot [color=green, mark=triangle] coordinates
			{
				(12,4.995)
				(16,3.997)
				(20,4.424)
				(24,19.903)
				(28,32.548)
				(32,56.417)   
			};
			\addplot [color=blue, mark=o] coordinates
			{
				(12,20.203)
				(16,75.741)
				(20,136.387)
				(24,187.043)
				(28,259.342)
				(32,300.031)
			};
			\addplot [color=black,mark=*] coordinates
			{
				(12,0.105)
				(16,2.036)
				(20,53.195)
				(24,300.406)
				(28,517.544)
				(32,933.722)
				
			};
			\addplot [color=black,mark=square] coordinates
			{
				(12,0.107)
				(16,1.957)
				(20,49.144)
				(24,290.205)
				(28,499.454)
				(32,912.527)
			};
			\tiny{\legend{\sctl{}, \tool{iProver Modulo}, \verds{}, \nusmv{}, \nuxmv{}}}
			\end{axis}
			\end{tikzpicture}
		\end{tabular}
	\end{subfigure}
	
	\caption{Average verification time in benchmark {\#}2.}
	\label{fig:average_time:extended}
	
	\begin{subfigure}\centering
		\begin{tabular}{c}
			
			\begin{tikzpicture}[scale=0.55]
			\begin{axis}[title={\small{CP}},legend pos=north west, small,
			xlabel = {Number of state variables},
			ylabel = {Memory [MB]}
			]
			\addplot [color=red, mark=x] coordinates
			{
				(12,7.845)
				(24,22.328)
				(36,42.184)  
				(48,55.100)
				(60,77.357)   
			};
			\addplot [color=green, mark=triangle] coordinates
			{
				(12,10.111)
				(24,16.547)
				(36,73.946)
				(48,123.342)
				(60,204.298)
			};
			\addplot [color=blue, mark=o] coordinates
			{
				(12,322.020)
				(24,468.169)
				(36,581.011)
				(48,601.023)
				(60,631.034)
			};
			\addplot [color=black,mark=*] coordinates
			{
				(12,8.818)
				(24,42.924)
				(36,251.364)
				(48,589.205)
				(60,1559.283)
			};
			\addplot [color=black,mark=square] coordinates
			{
				(12,21.013)
				(24,55.179)
				(36,256.058)
				(48,650.324)
				(60,1595.689)
			};
			\legend{\sctl{}, \tool{iProver Modulo}, \verds{}, \nusmv{}, \nuxmv{}}
			\end{axis}
			\end{tikzpicture}
			
		\end{tabular}
	\end{subfigure}
	\begin{subfigure}
		\centering
		\begin{tabular}{c}
			\scriptsize
			\begin{tikzpicture}[scale=0.55]
			\begin{axis}[title={\small{CSP}},legend pos=north west, small,
			xlabel = {Number of state variables},
			ylabel = {Memory [MB]}
			]
			\addplot [color=red, mark=x] coordinates
			{
				(12,1.984)
				(16,2.039)
				(20,3.383) 
				(24,9.903)
				(28,22.548)
				(32,36.417)     
			};
			\addplot [color=green, mark=triangle] coordinates
			{
				(12,10.070)
				(16,11.449)
				(20,23.660)
				(24,39.903)
				(28,52.548)
				(32,86.417)
			};
			\addplot [color=blue, mark=o] coordinates
			{
				(12,322.023)
				(16,485.081)
				(20,514.027)
				(24,530.238)
				(28,542.231)
				(32,580.357)
			};
			\addplot [color=black,mark=*] coordinates
			{
				(12,7.051)
				(16,29.151)
				(20,144.974)
				(24,420.406)
				(28,1217.544)
				(32,2903.722)
				
			};
			\addplot [color=black,mark=square] coordinates
			{
				(12,21.168)
				(16,48.423)
				(20,157.113)
				(24,490.205)
				(28,1296.454)
				(32,2932.527)
			};
			\tiny{\legend{\sctl{}, \tool{iProver Modulo}, \verds{}, \nusmv{}, \nuxmv{}}}
			\end{axis}
			\end{tikzpicture}
		\end{tabular}
	\end{subfigure}
	
	\caption{Average memory usage in benchmark {\#}2.}
	\label{fig:average_memory:extended}
\end{figure}

\subsubsection{Continuation vs. recursion.}
To show the importance of using continuation-passing style,
we have implemented a recursive version of our tool and compared the time
efficiency. In benchmark {\#}1, {\#}2, and {\#}3, \sctl{} solves about 10\% more test cases than \sctlprovr{}, and it outperforms \sctlprovr{} in almost
all solvable cases (Table~\ref{tabl:cont_vs_rec}). \sctlprovr{} is more sensitive to the number of variables than \sctl{} (Figure \ref{fig:average_time:recursive:vs:continuation}).

{\tiny
	\begin{table}\scriptsize
		\setlength{\tabcolsep}{1pt}
		\begin{center}
			\begin{tabular}{| l | r | r | r |}
				
				\hline
				\textbf{Bench} & \sctl{} solvable & \sctlprovr{} solvable & t(\sctlprov) $<$ t(\sctlprovr{})  \\
				\hline
				\textbf{\#1} & 2862(99.4\%) & 2682(93.1\%) & 2598(90.2\%)\\
				\hline
				\textbf{\#2} & 2597(90.2\%) & 2306(80.1\%) & 2406(83.5\%)\\
				\hline
				\textbf{\#3} & 1849(64.2\%) & 1520(52.8\%) & 1735(60.2\%)\\
				\hline
			\end{tabular}
		\end{center}
		\caption{\sctl{} vs. \sctlprovr{}}
		\label{tabl:cont_vs_rec}
	\end{table}
}

\begin{figure}[h]\tiny\centering
	\begin{subfigure}\centering
		\begin{tabular}{c}
			
			\begin{tikzpicture}[scale=0.55]
			\begin{axis}[title={\small{CP}},legend pos=north west, small,
			xlabel = {Number of state variables},
			ylabel = {Time [seconds]}
			]
			\addplot [color=red, mark=x] coordinates
			{
				(12,0.011)
				(24,0.280)
				(36,2.929)  
				(48,5.100)
				(60,7.357)
				(72,19.566)   
			};
			\addplot [color=black,mark=*] coordinates
			{
				(12,0.032)
				(24,2.238)
				(36,6.717)
				(48,17.578)
				(60,55.582)
				(72,101.265)
			};
			\legend{\sctlprov, \sctlprovr{}}
			\end{axis}
			\end{tikzpicture}
			
		\end{tabular}
	\end{subfigure}
	\begin{subfigure}
		\centering
		\begin{tabular}{c}
			\scriptsize
			\begin{tikzpicture}[scale=0.55]
			\begin{axis}[title={\small{CSP}},legend pos=north west, small,
			xlabel = {Number of state variables},
			ylabel = {Time [seconds]}
			]
			\addplot [color=red, mark=x] coordinates
			{
				(12,0.006)
				(16,0.007)
				(20,0.374) 
				(24,9.903)
				(28,12.548)
				(32,26.417)
				(52,91.134)
				(72,180.098)     
			};
			
			\addplot [color=black,mark=*] coordinates
			{
				(12,0.035)
				(16,1.238)
				(20,10.717)
				(24,30.406)
				(28,57.544)
				(32,83.722)
				(52,234.546)
				(72,504.256) 
				
			};
			\tiny{\legend{\sctlprov,\sctlprovr{}}}
			\end{axis}
			\end{tikzpicture}
		\end{tabular}
	\end{subfigure}
	
	\caption{Average verification time in \sctl{} vs. \sctlprovr.}
	\label{fig:average_time:recursive:vs:continuation}
\end{figure}

\begin{remark}
	In the comparison of average verification time of \sctl{} and \sctlprovr, we extend the number of variables in Concurrent Sequential Processes to 72, which is the same as in Concurrent Processes.
\end{remark}

\subsection{Programs with fairness constraints}\label{subsec:fair}
In this part, we evaluate benchmark \#4, which models mutual exclusion
algorithms and ring
algorithms\footnote{\url{http://lcs.ios.ac.cn/~zwh/verds/verds_code/bp12.rar}}.
Then, we compare the evaluation results of \sctl{}, \verds{},
\nusmv{}, and \nuxmv{}, and we do not consider \tool{iProver Modulo}
because \tool{iProver Modulo} cannot handle \CTL{} properties with
fairness constraints \cite{Ji15}.

\paragraph{Mutual exclusion and ring algorithms.}

This benchmark consists of two sets of concurrent programs: the mutual
exclusion algorithms and the ring algorithms. Both kinds of algorithms
consist of a set of concurrent processes running in parallel.

In the mutual exclusion algorithms, the scheduling of processes is simple: for all $i$ between $0$ and $n-2$, process $i+1$ performs a transition after process $i$, and process $0$ performs a transition after process $n-1$.
Each formula in the algorithms needs to be
verified under the fairness constraint that each process does not
starve, i.e., no process waits infinitely long.

Each process in the mutual exclusion algorithms has three internal
states: \textsf{noncritical}, \textsf{trying}, and
\textsf{critical}. The number of processes vary from $6$ to $51$. There
are five properties specified by \CTL{} formulae are to be verified in
mutual exclusion algorithms, as in
Table~\ref{tabl:mutual:ring:properties}. In these formulae, $non_i$
($try_i$, $cri_i$) indicates that process $p_i$ has internal state
\textsf{noncritical} (\textsf{trying}, \textsf{critical}).
Note that because of the scheduling algorithm, 
processes $0$ and $1$ are not symmetric, as exemplified by the 
difference in performance between the properties $P_4$ and $P_5$.

Each process in the ring algorithms consists of $5$ Boolean internal variables indicating the internal state, and a Boolean variable indicating the output. Each process receives a Boolean value as the input during its running time. For a ring algorithm with processes $p_0,p_1,...,p_n$, the internal state of $p_i$ depends on the output of process $p_{i-1}$, and the output of $p_{i-1}$ depends on its internal state, where $1\le i\le n$. The internal state of $p_0$ depends on the output of process $p_n$, and the output of $p_n$ depends on the internal state of its own. The number of processes vary from $3$ to $10$. There are four properties specified by \CTL{} formulae are to be verified in ring algorithms, as in Figure~\ref{tabl:mutual:ring:properties}. In these formulae, $out_i$ indicates that the output of process $p_i$ is Boolean value $true$.


The experimental results (Table~\ref{tabl:solvable:mutual:ring} and Table \ref{tabl:compare:mutual:ring}) show that \sctl{} solves more test cases than \verds, \nusmv{}, and \nuxmv{}. At the same time, \sctl{} is more time and space efficiency in more than 75 percent of the test cases than the other three tools.  

The detailed experimental data is shown in ~\ref{bench4:data:detail}.

\begin{table}[h!]
	\scriptsize
	\begin{center}
		\begin{tabular}{| l | l |}
			\hline
			\textbf{Prop} & \textbf{Mutual Exclusion Algorithms}\\
			\hline
			{$P_1$} & $EF (cri_0 \wedge cri_1)$  \\
			\hline
			{$P_2$} &  $AG (try_0 \Rightarrow AF (cri_0))$\\
			\hline
			{$P_3$} &  $AG (try_1 \Rightarrow AF (cri_1))$\\
			
			\hline
			{$P_4$} &  $AG (cri_0 \Rightarrow A cri_0 U (\neg cri_0 \wedge A \neg cri_0 U cri_1))$  \\
			\hline
			{$P_5$} &  $AG (cri_1 \Rightarrow A cri_1 U (\neg cri_1 \wedge A \neg cri_1 U cri_0))$\\
			\hline
		\hline
		\textbf{Prop} & \textbf{Ring Algorithms}\\
		\hline
		{$P_1$} & $AGAF out_0 \wedge AGAF \neg out_0$ \\
		\hline
		{$P_2$} &  $AGEF out_0 \wedge AGEF \neg out_0$ \\
		\hline
		{$P_3$} &  $EGAF out_0 \wedge EGAF \neg out_0$\\
		
		\hline
		{$P_4$} &  $EGEF out_0 \wedge EGEF \neg out_0$ \\
		\hline
	\end{tabular}
	\end{center}
	\figcaption{Properties to be verified in benchmark \#4.}
	\label{tabl:mutual:ring:properties}
\end{table}

\begin{table}[h!]
	\scriptsize
			\centering
			\begin{tabular}{| l | r | r | r | r |}
				\hline
				\textbf{Programs} & \verds{} & \nusmv{} & \nuxmv{} &  \sctl{} \\
				\hline
				\code{mutual exclusion} & 136 (59.1\%) & 50 (21.7\%) & 50 (21.7\%) & 191 (83.0\%)  \\
				\hline
				\code{ring} & 16 (50.0\%) & 21 (65.6\%) & 21 (65.6\%) & 20 (62.5\%) \\
				\hline
				Sum & 152(58.0\%) & 71(27.1\%) & 71(27.1\%) & 211 (80.5\%)\\
				\hline
			\end{tabular}	
			\caption{Solvable cases in \verds{}, \nusmv{}, \nuxmv{}, and \sctl{}.}
			\label{tabl:solvable:mutual:ring}
			\begin{tabular}{| l | r | r | r |}
				\hline
				\textbf{Programs} & \verds{} & \nusmv{} & \nuxmv{}  \\
				\hline
				\code{mutual exclusion} & 187 (81.3\%) & 191 (83.0\%) & 191 (83.0\%)   \\
				\hline
				\code{ring} & 13 (40.6\%) & 20 (62.5\%) & 20 (62.5\%)  \\
				\hline
				Sum & 200(76.3\%) & 211(80.5\%) & 211(80.5\%) \\
				\hline
			\end{tabular}
			\caption{Cases where \sctl{} both runs faster and uses less memory.}
			\label{tabl:compare:mutual:ring}
\end{table}

\subsection{Discussion of the experimental results}
In the evaluation of all benchmarks in this paper, the performances of the five tools in the comparisons are affected by two factors: the number of state variables, and the type of the property to be checked. The performances of \nusmv{} and \nuxmv{} are mainly affected by the number of state variables, while the performances of \tool{iProver Modulo}, \verds{}, and \sctl{} are mainly affected by the type of the property to be checked. When the number of state variables is rather small (such as test cases in benchmark \#1), \nusmv{} and \nuxmv{} solves more test cases than \tool{iProver Modulo}, \verds{} and \sctl{}, but when the number of state variables becomes larger (such as test cases in benchmark \#2 and \#3), they performs worse than the other three tools.  When checking properties where nearly all states must be searched (such as $AG$ properties), \nusmv{} and \nuxmv{} usually perform better than \tool{iProver Modulo}, \verds{} and \sctl{}. However, for most properties, \tool{iProver Modulo}, \verds{} and \sctl{} usually search much less states than \nusmv{} and \nuxmv{} to check them, and are more time and space efficiency. Thus, \tool{iProver Modulo}, \verds{} and \sctl{} scale up better than \nusmv{} and \nuxmv{} when checking these properties. Moreover, \sctl{} scales up better than both \tool{iProver Modulo} and \verds{}, and outperforms these two tools in most solvable cases.

\subsection{An application to the analysis of Air traffic control protocols}

As an application to an engineering problem, we present
a concept of operations for the {\em Small Aircraft Transportation System} 
(\textsf{SATS}) 
\cite{MunozDC04,nasasats04}
in \sctl{}\footnote{\url{https://github.com/terminatorlxj/SATS-model}}.

In this concept of operation, the airspace volume surrounding an
airport facility, called {\em the self controlled area}, is divided into
15 zones (Figure \ref{fig:example:sats:sca}).
\begin{figure}
	\centering
	\includegraphics[width=6cm]{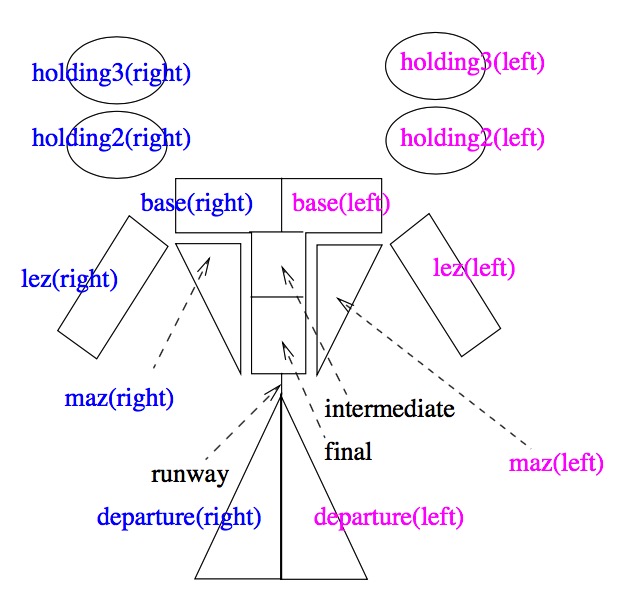}
	\caption{SCA zones, where right and left are relative to the pilot facing the runway, i.e., opposite from the reader point of view \cite{MunozDC04}.}
	\label{fig:example:sats:sca}
\end{figure}
For instance, the zone {\tt holding3(right)} is a holding pattern at
3000 feet on the right of the self controlled area.
Each zone contains a list of aircraft and 24 transition rules
specify different SATS-procedures.
For instance, the rule \textbf{Vertical Entry (right)}
specifies the vertical entry of an aircraft in the zone
{\tt holding3(right)}. 

The model is
non-deterministic, that is, for a given state, several transitions are
possible and all must be considered.  As there are no a
priori bounds on the number of aircraft in each zone, the number of
states in the model is potentially infinite. However, the number of
states that are reachable from the initial state is finite: 
an enumeration of the model shows that there are 54221 such states (and
around 3000 in the simplified model where departure operations are not
considered).

There are eight properties of the model that we want to 
verify with \sctl{}, for instance that 
the \textsf{SATS} concept does not allow more than four simultaneous 
landing operations and none of the 
15 zones contains too many aircraft (each zone is
assigned a maximum number of aircraft and the actual number of
aircraft is never higher than this number).
The safety property is thus conjunction of these eight properties. 

The verification problem is to check that this
property holds on every reachable state from the initial state (the
state where there are no aircraft on each zone of the self controlled
area), so the formula to be checked is $AG_x(\phi)(e)$ where $\phi$ is 
the conjunction of the eight properties and $e$ is the initial state.

This is a typical model checking problem, but this problem
is known to be cumbersome for traditional model checkers
\cite{MunozDC04} because:
\begin{itemize}
	\item Each state of the model is represented by a complex data
	structure. For instance, a number of state variables are
	represented by lists of aircraft with unbounded length.
	\item The transition rules of the model are 
	complex algorithms. For instance, some transitions rules
	involve recursive operations on lists of aircraft.
	\item The
	properties to be verified in the model are also represented
	by complex algorithms. For instance, some of the properties
	are inductively defined over lists of aircraft.
\end{itemize}

However, this example fits well in \sctl{} that provides a
more expressive input language than most traditional model
checkers. Indeed, \sctl{} provides both readable notations for the
definition of data structures such as records or lists with unbounded
length, and arbitrary algorithms for the definitions of transition
rules and of properties.  So we have been able to check in
\sctl{} that the safety property holds on the model, and the
verification was executed in less than 30 seconds on the same machine as which the benchmarks are evaluated.

\section{Conclusion and future work}\label{conclusion}

This paper provides a first step towards combining model checking
and proof checking.  

We proposed a parameterized logic \CTLP{},
which extends \CTL{} with polyadic predicate symbols, 
provided a proof system \SCTL{} for \CTLP{} in the style of a sequent calculus, 
and developed a
new automated theorem prover \sctl{} from scratch, tailored
for \SCTL{}.  The particular aspects of \sctl{} are as
follows: (1) It performs verification automatically and directly over
any given Kripke model.  (2) In addition of generating counterexamples
when the verification of the given property fails, \sctl{}
permits to give a certificate for the property when it succeeds.  (3)
It performs verification in a continuation-passing style and a doubly
on-the-fly style, thanks to the syntax and inference rules of
\SCTL{}.

As comparisons to other \CTL{} model checking tools, we consider four
other tools: an automated theorem prover \tool{iProver Modulo}, a
\QBF{}-based bounded model checker \verds{}, and two \BDD{}-based
symbolic model checker \nusmv{} and its extension \nuxmv{}. There are
four benchmarks considered in the comparisons. Benchmark \#1 is
originally introduced by Zhang \cite{Zhang14} in the evaluation of
\verds{} and \nusmv{}. Later, Ji \cite{Ji15} also uses this benchmark
in the evaluation of \tool{iProver Modulo} and \verds{}. Based on
benchmark \#1, we extend the number of state variables into tens,
hundreds, and even thousands in benchmark \#2 and benchmark \#3. In
benchmark \#4, we verify properties with fairness constraints on test
cases that models mutual exclusion algorithms and ring algorithms.
The experimental results show that \sctl{} has a good performance
in terms of time and space consuming, compared to existing tools and,
\sctl{} outperforms these four tools in the verification of many
kinds of \CTL{} properties, and can be considered complementary to
model checkers such as \nusmv{} and \nuxmv{}, which are among the best
\CTL{} model checkers up-to-date.  As a matter of fact, \nusmv{}
and \nuxmv{} perform better than \sctl{} in proving some
AG properties, while \sctl{} usually performs better with
other kinds of properties. Note also that the tool \sctl{} can
be seen either as a theorem prover, or a model checker that can
produce more information than traditional ones.

The fairness constraints have been added in the implementation, but not
yet in \CTLP{} nor in \SCTL{}. This is left for future work.

Until now, \sctl{} is single-threaded, it is also our future work to write a parallel version to improve efficiency.

\section*{Acknowledgment}
{This work is supported by the ANR-NSFC project LOCALI (NSFC 61161130530 and ANR 11 IS02 002 01).}

\bibliography{references} \bibliographystyle{splncs03}

\begin{thebibliography}{10}
\providecommand{\url}[1]{\texttt{#1}}
\providecommand{\urlprefix}{URL }

\bibitem{Appel06}
Appel, A.W.: Compiling with Continuations (corr. version). Cambridge University
  Press, UK (2006)

\bibitem{BaierKatoen08}
Baier, C., Katoen, J.: Principles of model checking. {MIT} Press, USA (2008)

\bibitem{BCG95}
Bhat, G., Cleaveland, R., Grumberg, O.: Efficient on-the-fly model checking for
  $ctl^*$. In: Proceedings of LICS'95. pp. 388--397. {IEEE} Computer Society,
  USA, San Diego, California, USA (June 26-29 1995)

\bibitem{BCCZ99}
Biere, A., Cimatti, A., Clarke, E., Zhu, Y.: Symbolic model checking without
  {BDD}s. In: Cleaveland, W.R. (ed.) Proceedings of TACAS'99. LNCS, vol. 1579,
  pp. 193--207. Springer, USA, Amsterdam, the Netherlands (March 20-28 1999)

\bibitem{BiereCCSZ03}
Biere, A., Cimatti, A., Clarke, E.M., Strichman, O., Zhu, Y.: Bounded model
  checking. Advances in Computers  58,  117--148 (2003)

\bibitem{BouajjaniJNT00}
Bouajjani, A., Jonsson, B., Nilsson, M., Touili, T.: Regular model checking.
  In: Proceedings of Computer Aided Verification, 12th International
  Conference, {CAV} 2000. pp. 403--418. Springer-Verlag, Berlin, Chicago, IL,
  USA (July 15-19 2000)

\bibitem{BruennlerLange}
Br{\"{u}}nnler, K., Lange, M.: Cut-free sequent systems for temporal logic. J.
  Log. Algebr. Program.  76(2),  216--225 (2008)

\bibitem{Burel09}
Burel, G.: Automating theories in intuitionistic logic. In: Proceedings of
  Frontiers of Combining Systems, 7th International Symposium, FroCoS 2009. pp.
  181--197. Springer-Verlag, Berlin, Trento, Italy (September 16-18 2009)

\bibitem{Burel11}
Burel, G.: Experimenting with deduction modulo. In: Sofronie-Stokkermans, V.,
  Bj{\o}rner, N. (eds.) Proceedings of CADE 2011. pp. 162--176.
  Springer-Verlag, Berlin, Wroclaw, Poland (July 31-August 5 2011)

\bibitem{CAVCDGMMMRT14}
Cavada, R., Cimatti, A., Dorigatti, M., Griggio, A., Mariotti, A., Micheli, A.,
  Mover, S., Roveri, M., Tonetta, S.: The nuxmv symbolic model checker. In:
  Proceedings of Computer Aided Verification - 26th International Conference,
  {CAV} 2014. pp. 334--342. Springer International Publishing, Switzerland,
  Vienna, Austria (July 18-22 2014)

\bibitem{CimattiCGR99}
Cimatti, A., Clarke, E.M., Giunchiglia, F., Roveri, M.: Nusmv: A new symbolic
  model verifier. In: Proceedings of CAV'99. pp. 495--499. Springer-Verlag,
  Berlin, Trento, Italy (July 6-10 1999)

\bibitem{CGP01}
Clarke, E.M., Grumberg, O., Peled, D.: Model checking. {MIT} Press, Cambridge,
  MA, USA (2001)

\bibitem{Craig89}
Craig, J.J.: Introduction to robotics - mechanics and control {(2.} ed.).
  Prentice Hall, USA (1989)

\bibitem{ND82}
Dershowitz, N.: Orderings for term-rewriting systems. Theor. Comput. Sci.  17,
  279--301 (1982)

\bibitem{EmersonC82}
Emerson, E.A., Clarke, E.M.: Using branching time temporal logic to synthesize
  synchronization skeletons. Sci. Comput. Program.  2(3),  241--266 (1982)

\bibitem{EmersonH85}
Emerson, E.A., Halpern, J.Y.: Decision procedures and expressiveness in the
  temporal logic of branching time. J. Comput. Syst. Sci.  30(1),  1--24 (1985)

\bibitem{FisherDP01}
Fisher, M., Dixon, C., Peim, M.: Clausal temporal resolution. {ACM} Trans.
  Comput. Log.  2(1),  12--56 (2001)

\bibitem{Fitting96}
Fitting, M.: First-Order Logic and Automated Theorem Proving, Second Edition.
  Graduate Texts in Computer Science, Springer-Verlag, New York (1996)

\bibitem{GabbayP08}
Gabbay, D.M., Pnueli, A.: A sound and complete deductive system for ctl*
  verification. Logic JOURNAL of the {IGPL}  16(6),  499--536 (2008)

\bibitem{Ji15}
Ji, K.: {CTL Model Checking in Deduction Modulo}. In: Proceedings of Automated
  Deduction - CADE-25. pp. 295--310. Springer International Publishing,
  Switzerland, Berlin (August 1-7 2015)

\bibitem{Loveland78}
Loveland, D.W.: Automated Theorem Proving: A Logical Basis (Fundamental Studies
  in Computer Science). Elsevier, Amsterdam (1978)

\bibitem{mcmillan93}
McMillan, K.L.: Symbolic Model checking. Springer, USA (1993)

\bibitem{MunozDC04}
Mu{\~{n}}oz, C.A., Dowek, G., Carre{\~{n}}o, V.: Modeling and verification of
  an air traffic concept of operations. In: Proceedings of the {ACM/SIGSOFT}
  International Symposium on Software Testing and Analysis, {ISSTA} 2004. pp.
  175--182. ACM, USA, Boston, Massachusetts, USA (July 11-14 2004)

\bibitem{nasasats04}
NASA/TM-2004-213006: Abstract Model of SATS Concept of Operations: Initial
  Results and Recommendations. NASA, USA (2004)

\bibitem{PartoviL14}
Partovi, A., Lin, H.: Assume-guarantee cooperative satisfaction of multi-agent
  systems. In: Proceedings of American Control Conference, {ACC} 2014. pp.
  2053--2058. IEEE, USA, USA (June 4-6 2014)

\bibitem{Peterson81}
Peterson, G.L.: Myths about the mutual exclusion problem. Inf. Process. Lett.
  12(3),  115--116 (1981)

\bibitem{PnueliK02}
Pnueli, A., Kesten, Y.: A deductive proof system for {CTL}. In: Proceedings of
  {CONCUR} 2002. pp. 24--40. Springer-Verlag, Berlin, Brno, Czech Republic
  (August 20-23 2002)

\bibitem{Reynolds93}
Reynolds, J.C.: The discoveries of continuations. Lisp and Symbolic Computation
   6(3-4),  233--248 (1993)

\bibitem{Reynolds01}
Reynolds, M.: An axiomatization of full computation tree logic. J. Symb. Log.
  66(3),  1011--1057 (2001)

\bibitem{Sestoft12}
Sestoft, P.: Programming Language Concepts, Undergraduate Topics in Computer
  Science, vol.~50. Springer International Publishing, Switzerland (2012)

\bibitem{VergauwenL93}
Vergauwen, B., Lewi, J.: A linear local model checking algorithm for {CTL}. In:
  Proceedings of {CONCUR} '93, 4th International Conference on Concurrency
  Theory. pp. 447--461. Springer-Verlag, Berlin, Hildesheim, Germany (August
  23-26 1993)

\bibitem{Zhang14}
Zhang, W.: {QBF Encoding of Temporal Properties and QBF-based Verification}.
  In: Proceedings of IJCAR 2014. pp. 224--239. Springer-Verlag, Berlin, Vienna
  (July 19-22 2014)

\end{thebibliography}
\section*{Appendix}
\appendix

\section{Detailed Experimental data in benchmark \#1, \#2 and \#3}\label{app:detail:data}
We show the detailed experimental data in benchmark \#1, \#2 and \#3 in the following three subsections.
\subsection{Benchmark \#1 (Table~\ref{tabl:solvable} and \ref{tabl:compare})}

Table \ref{tabl:solvable} shows that \sctl{} outperforms \tool{iProver Modulo}
and \verds, and is almost as good as \nusmv{} and \nuxmv{}: \nusmv{} and
\nuxmv{} solve all the 2880 problems, while \sctl{} solves 2862
problems (99.4\%).

Let us now turn to the efficiency. \sctl{} is much faster than the
four other tools (Table \ref{tabl:compare}). Among the problems that can be solved by \sctl{} and
\tool{iProver Modulo}, \sctl{} is faster in 98.2\% of these
problems, 99.2\% when compared with \verds, 95.2\% when compared with
\nusmv{} and 95.9\% when compared with \nuxmv.

\begin{figure}[h]\scriptsize
	\centering
	\setlength{\tabcolsep}{1pt}
	\begin{tabular}{| l | r | r | r | r | r |}
		\hline
		\textbf{Programs} & \tool{iProver Modulo} & \verds{} & \nusmv{} & \nuxmv{} & \sctl{}\\
		\hline
		\code{CP ($b = 12$)} & 467(97.3\%) & 433(90.2\%) & 480(100\%) & 480(100\%) & 480(100\%)\\
		\hline
		
		\code{CP ($b = 24$)} & 372(77.5\%) & 428(89.2\%) & 480(100\%) & 480(100\%) & 480(100\%)\\
		\hline
		
		\code{CP ($b = 36$)} & 383(79.8\%) & 416(86.7\%) & 480(100\%) & 480(100\%) & 470(97.9\%)\\
		\hline
		
		\code{CSP ($b = 12$)} & 177(36.9\%) & 370(77.1\%) & 480(100\%) & 480(100\%) & 480(100\%)\\
		\hline
		
		\code{CSP ($b = 16$)} & 164(34.2\%) & 315(65.6\%) & 480(100\%) & 480(100\%) & 474(98.8\%)\\
		\hline
		
		\code{CSP ($b = 20$)} & 253(52.7\%) & 268(55.8\%) & 480(100\%) & 480(100\%) & 478(99.6\%)\\
		\hline
		Sum & 1816(63.1\%) & 2230(77.4\%) & 2880(100\%) & 2880(100\%) & 2862(99.4\%)\\
		\hline
	\end{tabular}
	\tabcaption{Solvable cases in five tools.}
	\label{tabl:solvable}
	\vspace{0.5cm}
	\begin{tabular}{| l | r | r | r | r |}
		\hline
		\textbf{Programs} & \tool{iProver Modulo} & \verds{} & \nusmv{} & \nuxmv{} \\
		\hline
		\code{CP ($b = 12$)} & 480(100\%) & 480(100\%) & 430(89.6\%) & 431(89.8\%) \\
		\hline
		\code{CP ($b = 24$)} & 480(100\%) & 480(100\%) & 456(95.0\%) & 458(95.4\%) \\
		\hline
		\code{CP ($b = 36$)} & 454(94.6\%) & 467(97.3\%) & 441(91.9\%) & 446(92.9\%) \\
		\hline
		\code{CSP ($b = 12$)} & 480(100\%) & 480(100\%) & 464(96.7\%) & 465(96.9\%) \\
		\hline
		\code{CSP ($b = 16$)} & 474(98.6\%) & 473(98.5\%) & 472(98.3\%) & 474(98.6\%) \\
		\hline
		\code{CSP ($b = 20$)} & 455(94.8\%) & 478(99.6\%) & 478(99.6\%) & 479(99.8\%) \\
		\hline
		Sum & 2823(98.2\%) & 2858(99.2\%) & 2741(95.2\%) & 2763(95.9\%)  \\
		\hline
	\end{tabular}
	
	\tabcaption{Cases where \sctl{} runs faster.}
	\label{tabl:compare}
\end{figure}

\subsection{Benchmark \#2 (Table~\ref{tabl:solvable:extended} and \ref{tabl:compare:extended})}

Our benchmark {\#}2 investigates the performances of \tool{iProver Modulo}, \verds{}, \nusmv{}, \nuxmv{}, and \sctl{}
when the size of the model increases.

To do so, we increase the number of variables in the random Boolean
programs to $48$, $60$, or $72$ for concurrent processes, and $24$,
$28$, or $32$ for concurrent sequential processes.
The 2880 test cases are also randomly generated.
The properties to be checked are the same as in benchmark \#1. 

Counting the number of problems that can be solved in 20 minutes,
we see that \sctl{} scales up better (Table \ref{tabl:solvable:extended}, \ref{tabl:compare:extended}) than the other four tools: \sctl{} solves more test cases than the other tools and, outperforms the other tools in most solvable test cases.

\begin{figure}[h]\scriptsize
			\centering
			\setlength{\tabcolsep}{3pt}
			\begin{tabular}{| l | r | r | r | r | r |}
				\hline
				\textbf{Programs} & \tool{iProver Modulo} & \verds{} & \nuxmv{} & \nuxmv{} & \sctl{} \\
				\hline
				\code{CP ($b = 48$)} & 375(78.1\%) & 400(83.3\%) & 171(35.6\%) & 176(36.7\%) & 446(92.9\%)  \\
				\hline
				\code{CP ($b = 60$)} & 360(75.0\%) & 403(84.0\%) & 22(4.6\%) & 23(4.8\%) & 440(91.7\%)  \\
				\hline
				\code{CP ($b = 72$)} & 347(72.3\%) & 383(79.8\%) &  0 & 0 & 437(91.0\%)  \\
				
				\hline
				\code{CSP ($b=24$)} & 190(39.6\%) & 235(49.0\%) &  421(87.7\%) & 423(88.1\%) & 430(89.6\%) \\
				\hline
				\code{CSP ($b=28$)} & 172(35.8\%) & 229(47.7\%) & 106(22.1\%) & 108(22.5\%) & 426(88.8\%) \\
				\hline
				\code{CSP ($b=32$)} & 158(32.9\%) & 224(46.7\%) & 8(1.7\%) & 6(1.3\%) & 418(87.1\%) \\
				
				\hline
				Sum & 1602(55.6\%) & 1874(65.1\%) & 728(25.3\%) & 736(25.6\%) & 2597(90.2\%)\\
				\hline
			\end{tabular}
			\tabcaption{Solvable cases in four tools.}
			\label{tabl:solvable:extended}
			\vspace{0.5cm}
		
			\begin{tabular}{| l | r | r | r | r |}
				\hline
				\textbf{Programs} & \tool{iProver Modulo} & \verds{} & \nusmv{} & \nuxmv{}\\
				\hline
				\code{CP ($b=48$)} & 446(92.9\%) & 444(92.5\%) & 442(92.1\%) & 442(92.1\%) \\
				\hline
				\code{CP ($b=60$)} & 440(91.7\%) & 440(91.7\%) & 440(91.7\%) & 440(91.7\%) \\
				\hline
				\code{CP ($b=72$)} & 437(91.0\%) & 437(91.0\%) & 437(91.0\%) & 437(91.0\%) \\
				
				\hline
				\code{CSP ($b=24$)} & 430(89.6\%) & 429(89.4\%) & 426(88.8\%) & 426(88.8\%) \\
				\hline
				\code{CSP ($b=28$)} & 426(88.8\%) & 426(88.8\%) & 425(88.5\%) & 425(88.5\%) \\
				\hline
				\code{CSP ($b=32$)} & 418(87.1\%) & 418(87.1\%) & 418(87.1\%) & 418(87.1\%) \\
				\hline
				Sum & 2597(90.2\%) & 2594(90.1\%) & 2588(89.9\%) & 2588(89.9\%)\\
				\hline
			\end{tabular}	
			\tabcaption{Cases where \sctl{} runs faster.}
			\label{tabl:compare:extended}
\end{figure}
\subsection{Benchmark \#3 (Table~\ref{tabl:solvable:larger})}

We increase, in our benchmark {\#}3, the number of variables to $252$, $504$
and $1008$ for both concurrent and concurrent sequential processes.

We compare the evaluation results 
of \tool{iProver Modulo}, \verds{}, \nusmv, \nuxmv, and \sctl{} as in Table \ref{tabl:solvable:larger}, 
and find that, in 20 minutes, \sctl{} can still
solve 64.0\% of the test cases, while \tool{iProver Modulo} and \verds{} solve 39.8\% and 12.2\% test cases, respectively; moreover, \nusmv{} and \nuxmv{} solve none.

\begin{figure}[h]\scriptsize
	\setlength{\tabcolsep}{3pt}
	\begin{center}
		\begin{tabular}{| l | r | r | r | r | r |}
			\hline
			\textbf{Programs} & \tool{iProver Modulo} & \verds{} &
			\nusmv{} & \nuxmv{} & \sctl{} \\
			\hline
			\code{CP ($b=252$)} & 299(62.3\%) & 216(45.0\%) & 0 & 0 & 371(77.3\%) \\
			\hline
			\code{CP ($b=504$)} & 292(60.8\%) & 0 & 0 & 0 & 335(69.8\%)\\
			\hline
			\code{CP ($b=1008$)} & 271(56.5\%) & 0 & 0 & 0 & 278(57.9\%)\\
			
			\hline
			\code{CSP ($b=252$)} & 114(23.6\%) & 136(28.3\%) & 0 & 0 & 312(65.0\%) \\
			\hline
			\code{CSP ($b=504$)} & 108(22.5\%) & 0 & 0 & 0 & 295(61.5\%) \\
			\hline
			\code{CSP ($b=1008$)} & 62(12.9\%) & 0 & 0 & 0 & 253(52.7\%)\\
			\hline
			Sum & 1146(39.8\%) & 352(12.2\%) & 0 & 0 & 1844(64.0\%)\\ \hline
		\end{tabular}
	\end{center}
	\tabcaption{Solvable cases with variable number 252, 504, and 1008, respectively.}
	\label{tabl:solvable:larger}
\end{figure}

\section{Experimental data in benchmark \#4}\label{bench4:data:detail}
The detailed experimental data of verifying test cases in benchmark \#4 is depicted in Table~\ref{tabl:data:mutual} and Table~\ref{tabl:data:ring}.
\begin{figure}[h!]\scriptsize
	\centering
			\begin{tabular}{| r | r | r | r | r | r | r | r | r | r |}
				\hline
				\textbf{Prop} & \textbf{NoP} & \multicolumn{8}{c|}{Mutual Exclusion Algorithms} \\
				\hline
				{} & {} & \multicolumn{2}{c|}{\verds{}} & \multicolumn{2}{c|}{\nusmv{}} & \multicolumn{2}{c|}{\nuxmv{}} &  \multicolumn{2}{c|}{\sctl{}} \\
				\hline
				{} & {} & sec & MB & sec & MB & sec & MB &  sec & MB \\
				\hline
				\multirow{6}{*}{$P_1$} & 6 & 0.286 & 321.99 & 0.153 & 9.07 & 0.270 & 21.18 & 0.005 & 2.25 \\
				{} & 12 & 1.278 & 322.08 & 19.506 & 76.98 & 21.848 & 89.25 & 0.016 & 3.70  \\
				{} & 18 & 4.719 & 426.45 & - & - & - & - & 0.037 & 5.44  \\
				{} & 24 & 11.989 & 601.55 & - & - & - & - & 0.091 & 9.36  \\
				{} & 30 & 26.511 & 926.25 & - & - & - & - & 0.200 & 16.49  \\
				{} & 36 & 52.473 & 1287.57 & - & - & - & - & 0.418 & 27.46  \\
				{} & 42 & 100.071 & 1944.95 & - & - & - & - & 0.682 & 48.28  \\
				{} & 48 & - & - & - & - & - & - & 1.119 & 66.63  \\
				{} & 51 & - & - & - & - & - & - & 1.392 & 82.32  \\
				\hline
				\multirow{6}{*}{$P_2$} & 6 & 0.375 & 322.07 & 0.054 & 9.07 & 0.048 & 21.31 & 0.012 & 3.07  \\
				{} & 12 & 2.011 & 322.02 & 22.774 & 76.96 & 21.733 & 89.24 & 0.035 & 4.44  \\
				{} & 18 & 7.958 & 446.71 & - & - & - & - & 0.101 & 8.09  \\
				{} & 24 & 23.448 & 692.30 & - & - & - & - & 0.252 & 14.57  \\
				{} & 30 & 48.800 & 1026.48 & - & - & - & - & 0.509 & 23.61  \\
				{} & 36 & 105.183 & 1619.01 & - & - & - & - & 1.005 & 50.49  \\
				{} & 42 & - & - & - & - & - & - & 1.791 & 57.93  \\
				{} & 48 & - & - & - & - & - & - & 2.679 & 86.67  \\
				{} & 51 & - & - & - & - & - & - & 3.453 & 129.83  \\
				\hline
				\multirow{6}{*}{$P_3$} & 6 & 0.331 & 322.02 & 0.089 & 9.04 & 0.033 & 21.27 & 0.012 & 3.03  \\
				{} & 12 & 2.059 & 322.07 & 22.749 & 76.91 & 21.897 & 89.22 & 0.035 & 4.93  \\
				{} & 18 & 7.995 & 449.13 & - & - & - & - & 0.110 & 9.59  \\
				{} & 24 & 23.578 & 696.74 & - & - & - & - & 0.286 & 21.04  \\
				{} & 30 & 51.774 & 1138.27 & - & - & - & - & 0.643 & 30.09  \\
				{} & 36 & 106.027 & 1628.84 & - & - & - & - & 1.287 & 66.14  \\
				{} & 42 & - & - & - & - & - & - & 2.138 & 86.29  \\
				{} & 48 & - & - & - & - & - & - & 3.369 & 170.94  \\
				{} & 51 & - & - & - & - & - & - & 4.333 & 149.03  \\
				\hline
				\multirow{6}{*}{$P_4$} & 6 & 0.446 & 321.97 & 0.089 & 9.04 & 0.033 & 21.27 & 0.039 & 3.38  \\
				{} & 12 & 8.289 & 552.62 & 22.749 & 76.91 & 21.897 & 89.22 & 150.115 & 986.64  \\
				{} & 18 & - & - & - & - & - & - & - & -  \\
				{} & 24 & - & - & - & - & - & - & - & -  \\
				{} & 30 & - & - & - & - & - & - & - & -  \\
				{} & 36 & - & - & - & - & - & - & - & -  \\
				{} & 42 & - & - & - & - & - & - & - & -  \\
				{} & 48 & - & - & - & - & - & - & - & -  \\
				{} & 51 & - & - & - & - & - & - & - & -  \\
				\hline
				\multirow{6}{*}{$P_5$} & 6 & 0.430 & 322.03 & 0.031 & 9.09 & 0.047 & 21.19 & 0.011 & 3.10  \\
				{} & 12 & 3.398 & 363.78 & 22.747 & 77.01 & 22.029 & 89.17 & 0.040 & 4.81  \\
				{} & 18 & 18.176 & 783.24 & - & - & - & - & 0.115 & 10.99  \\
				{} & 24 & 87.432 & 2382.82 & - & - & - & - & 0.322 & 18.68  \\
				{} & 30 & - & - & - & - & - & - & 1.414 & 47.68  \\
				{} & 36 & - & - & - & - & - & - & 1.287 & 66.35  \\
				{} & 42 & - & - & - & - & - & - & 2.405 & 142.86  \\
				{} & 48 & - & - & - & - & - & - & 4.848 & 225.55  \\
				{} & 51 & - & - & - & - & - & - & 5.177 & 225.66  \\
				\hline
			\end{tabular}
		\vspace{0.3cm}
		\tabcaption{Time and memory usage in benchmark \#4 (Mutual exclusion algorithms).}
		\label{tabl:data:mutual}
	\end{figure}

			\begin{figure}[h!]
				\scriptsize
				\centering
			\begin{tabular}{| r | r | r | r | r | r | r | r | r | r |}
				\hline
				\textbf{Prop} & \textbf{NoP} & \multicolumn{8}{c|}{Ring Algorithms} \\
				\hline
				{} & {} & \multicolumn{2}{c|}{\verds{}} & \multicolumn{2}{c|}{\nusmv{}} & \multicolumn{2}{c|}{\nuxmv{}} &  \multicolumn{2}{c|}{\sctl{}} \\
				\hline
				{} & {} & sec & MB & sec & MB & sec & MB & sec & MB \\
				\hline
				\multirow{6}{*}{$P_1$} & 3 & 0.168 & 322.09 & 0.040 & 10.02 & 0.045 & 22.08 & 4.622 & 62.22  \\
				{} & 4 & 0.216 & 322.12 & 0.299 & 22.46 & 0.255 & 34.96 & - & -  \\
				{} & 5 & 0.301 & 322.07 & 2.421 & 59.31 & 1.195 & 71.53 & - & -  \\
				{} & 6 & 0.449 & 322.13 & 22.127 & 80.49 & 17.967 & 92.82 & - & -  \\
				{} & 7 & 0.740 & 322.19 & 147.895 & 224.17 & 131.735 & 236.50 & - & -  \\
				{} & 8 & 1.115 & 322.09 & 1135.882 & 865.04 & 1083.48 & 877.36 & - & -  \\
				{} & 9 & 1.646 & 322.07 & - & - & - & - & - & -  \\
				{} & 10 & 2.232 & 321.96 & - & - & - & - & - & -  \\
				\hline
				\multirow{6}{*}{$P_2$} & 3 & - & - & 0.058 & 10.74 & 0.068 & 22.73 & 0.031 & 3.22  \\
				{} & 4 & - & - & 0.583 & 40.29 & 0.562 & 52.61 & 0.125 & 3.73  \\
				{} & 5 & - & - & 5.164 & 62.29 & 5.295 & 74.62 & 0.444 & 4.05  \\
				{} & 6 & - & - & 39.085 & 81.85 & 37.969 & 93.96 & 1.373 & 4.71  \\
				{} & 7 & - & - & 246.123 & 229.07 & 241.375 & 241.15 & 3.745 & 6.03  \\
				{} & 8 & - & - & - & - & - & - & 9.154 & 7.61  \\
				{} & 9 & - & - & - & - & - & - & 19.997 & 10.07  \\
				{} & 10 & - & - & - & - & - & - & 40.331 & 13.05  \\
				\hline
				\multirow{6}{*}{$P_3$} & 3 & - & - & 0.045 & 10.03 & 0.071 & 22.32 & 0.022 & 3.20  \\
				{} & 4 & - & - & 0.296 & 22.46 & 0.299 & 34.96 & 0.820 & 13.11  \\
				{} & 5 & - & - & 2.357 & 59.31 & 2.526 & 71.63 & 111.96 & 676.29  \\
				{} & 6 & - & - & 22.147 & 80.49 & 21.304 & 92.93 & - & -  \\
				{} & 7 & - & - & 147.567 & 224.17 & 141.134 & 236.74 & - & -  \\
				{} & 8 & - & - & - & - & - & - & - & -  \\
				{} & 9 & - & - & - & - & - & - & - & -  \\
				{} & 10 & - & - & - & - & - & - & - & -  \\
				\hline
				\multirow{6}{*}{$P_4$} & 3 & 0.158 & 322.09 & 0.066 & 10.00 & 0.171 & 22.32 & 0.024 & 3.24  \\
				{} & 4 & 0.190 & 322.05 & 0.356 & 22.46 & 0.367 & 34.95 & 0.104 & 3.82  \\
				{} & 5 & 0.263 & 322.04 & 2.726 & 59.31 & 2.781 & 71.63 & 0.385 & 3.99  \\
				{} & 6 & 0.385 & 322.07 & 27.013 & 80.48 & 24.794 & 94.95 & 1.289 & 4.57  \\
				{} & 7 & 0.528 & 322.07 & 181.007 & 224.16 & 166.725 & 236.61 & 3.727 & 5.29  \\
				{} & 8 & 0.815 & 322.14 & - & - & - & - & 9.525 & 7.14  \\
				{} & 9 & 1.138 & 322.19 & - & - & - & - & 21.568 & 9.31  \\
				{} & 10 & 1.574 & 321.98 & - & - & - & - & 45.097 & 12.95  \\
				\hline
			\end{tabular}
		
	\tabcaption{Time and memory usage in benchmark \#4 (Ring algorithms).}
	\label{tabl:data:ring}
\end{figure}

\section{Proof of soundness and completeness of \SCTL{}}\label{appendix:proof:sound:complete}
Proposition~\ref{prop:fis} and \ref{prop:fit} below permit to transform finite structures into infinite ones and will be used in the Soundness proof, while Proposition~\ref{prop:ifs} and \ref{prop:ift} permit to transform infinite structures into finite ones and will be used in the Completeness proof.

\begin{proposition}[Finite to infinite sequences]\label{prop:fis}
	Let $s_0,...,s_n$ be a finite sequence of states such that for all $i$
	between $0$ and $n-1$, $s_i\longrightarrow s_{i+1}$, and $s_n=s_p$ for
	some $p$ between $0$ and $n-1$. Then there exists an infinite sequence
	of states $s_0',s_1',...$ such that 
$s_0 = s_0'$ and
for all $i$, $s_i'\longrightarrow
	s_{i+1}'$, and all the $s_j'$ are among $s_0,...,s_n$.
\end{proposition}
\begin{proof}
	Take the sequence $s_0,...,s_{p-1},s_p,...,s_{n-1},s_p,...$, where $s_0 = s_0'$.
	\label{proof:prop:fis}
\end{proof}

\begin{proposition}[Finite to possibly infinite trees]\label{prop:fit}
	Let $\Phi$ be a set of states and $T$ be a finite tree labeled by
	states such that, for each internal node $s$, the immediate successors
	of $s$ are the elements of $\textsf{Next}(s)$ and each leaf is labeled with a
	state which is either in $\Phi$ or also a label of a node on the
	branch from the root of $T$ to this leaf. Then, there exists an
possibly
infinite tree $T'$ labeled by states such that for each internal node
	$s$ the successors of $s$ are the elements of $\textsf{Next}(s)$, all the
	leaves are labeled by elements of $\Phi$, and all the labels of $T'$
	are the labels $T$.
\end{proposition}
\begin{proof}
	Consider for $T'$ the tree whose root is labeled by the root
	of $T$ and such that for each node $s$, if $s$ is in $\Phi$,
	then $s$ is a leaf of $T'$, otherwise the successors of $s$
	are the elements of $\textsf{Next}(s)$. It is easy to check that all
	the nodes of $T'$ are labeled by labels of $T$.
	\label{proof:prop:fit}
\end{proof}

\begin{proposition}[Infinite to finite sequences]\label{prop:ifs}
	Let $s_0,s_1,...$ be an infinite sequence of states such that for all
	$i$, $s_i\longrightarrow s_{i+1}$. Then there exists a finite sequence
	of states $s_0',...,s_n'$ such that for all $i$ between $0$ and $n-1$,
	$s_i'\longrightarrow s_{i+1}'$, $s_n'=s_p'$ for some $p$ between $0$
	and $n-1$, and all the $s_j'$ are among $s_0,s_1,...$
\end{proposition}
\begin{proof}
	As the number of states is finite, there exists $p$ and $n$ such that $p<n$ and $s_p=s_n$. Take the sequence $s_0,...,s_n$.
	\label{proof:prop:ifs}
\end{proof}

\begin{proposition}[Possibly infinite to finite trees]\label{prop:ift}
	Let $\Phi$ be a set of states and $T$ be an 
possibly
infinite tree labeled by
	states such that for each internal node $s$ the successors of $s$ are
	the elements of $\textsf{Next}(s)$ and each leaf is labeled by a state in
	$\Phi$. Then, there exists a finite tree labeled by states such that
	for each internal node $s$ the successors of $s$ are the elements of
	$\textsf{Next}(s)$ and each leaf is labeled with a state which is either in
	$\Phi$ or also a label of a node on the branch from the root of $T$ to
	this leaf.
\end{proposition}
\begin{proof}
	As the number of states is finite, on each infinite branch, there exists $p$ and $n$ such that $p<n$ and $s_p=s_n$. Prune the tree at node $s_n$. This tree is finitely branching and each branch is finite, hence, by K\"onig's lemma, it is finite.
	\label{proof:prop:ift}
\end{proof}

\begin{theorem}[Soundness]\label{thm:sound}
	Let $\phi$ be a closed formula. If the sequent $\vdash \phi$ has a
	proof $\pi$, then $\models \phi$.
\end{theorem}
\begin{proof}
	By induction on the structure of the proof $\pi$.
	\begin{itemize}
		\item If the last rule of $\pi$ is \textsf{atom-R}, then the proved sequent has the form $\vdash P(s_1,...,s_n)$, hence $\models P(s_1,...,s_n)$.
		\item If the last rule of $\pi$ is $\neg$-\textsf{R}, then the proved sequent has the form $\vdash \neg P(s_1,...,s_n)$, hence $\models \neg P(s_1,...,s_n)$.
		\item If the last rule of $\pi$ is $\top$-\textsf{R}, the proved sequent has the form $\vdash \top$ and hence $\models \top$.
		\item If the last rule of $\pi$ is $\wedge$-\textsf{R}, then the proved sequent has the form $\vdash \phi_1\wedge \phi_2$. By induction hypothesis $\models \phi_1$ and $\models \phi_2$, hence $\models \phi_1\wedge \phi_2$.
		\item If the last rule of $\pi$ is $\vee$-$\mathsf{R_1}$ or $\vee$-$\mathsf{R_2}$, then the proved sequent has the form $\vdash \phi_1\vee\phi_2$. By induction hypothesis $\models \phi_1$ or $\models \phi_2$, hence $\models \phi_1\vee\phi_2$.
		\item If the last rule of $\pi$ is $\mathbf{AX}$-\textsf{R}, then the proved sequent has the form $\vdash AX_x(\phi_1)(s)$. By induction hypothesis, for each $s'$ in $\textsf{Next}(s)$, such that $\models (s'/x)\phi_1$, hence $\models AX_x(\phi_1)(s)$.
		\item If the last rule of $\pi$ is $\mathbf{EX}$-\textsf{R}, then the proved sequent has the form $\vdash EX_x(\phi_1)(s)$. By induction hypothesis, for each $s'$ in $\textsf{Next}(s)$, $\models (s'/x)\phi_1$, hence $\models EX_x(\phi_1)(s)$.
		\item If the last rule of $\pi$ is $\mathbf{AF}$-$\mathsf{R_1}$ or $\mathbf{AF}$-$\mathsf{R_2}$, then the proved sequent has the form $\vdash AF_x(\phi_1)(s)$. We associate a finite tree $|\pi|$ to the proof $\pi$ by induction in the following way.
		\begin{itemize}
			\item If the proof $\pi$ ends with the $\mathbf{AF}$-$\mathsf{R_1}$ rule with a subproof $\rho$ of the sequent $\vdash (s/x)\phi_1$, then the tree contains a single node $s$.
			\item IF the proof $\pi$ ends with the $\mathbf{AF}$-$\mathsf{R_2}$ rule, with subproofs $\pi_1,...,\pi_n$ of the sequent \\$\vdash AF_x(\phi_1)(s_1),...,\vdash AF_x(\phi_1)(s_n)$, respectively, then $|\pi|$ is the tree $s(|\pi_1|,...,|\pi_n|)$.
		\end{itemize}
		The tree $|\pi|$ has root $s$; for each internal node $s'$, the children of this node are labeled by elements of $\textsf{Next}(s')$; and for each leaf $s'$ the sequent $\vdash (s'/x)\phi_1$ has a proof smaller than $\pi$. By induction hypothesis, for each leaf $s'$ of $|\pi|$, $\models (s'/x)\phi_1$. Hence $\models AF_x(\phi_1)(s)$.
		\item If the last rule of $\pi$ is $\mathbf{EG}$-\textsf{R}, then the proved sequent has the form $\vdash EG_x(\phi_1)(s)$. We associate a finite sequence $|\pi|$ to the proof $\pi$ by induction in the following way.
		\begin{itemize}
			\item If the proof $\pi$ ends with the $\mathbf{EG}$-merge rule, then the sequence contains a single element $s$.
			\item If the proof $\pi$ ends with the $\mathbf{EG}$-\textsf{R} rule, with subproofs $\rho$ and $\pi_1$ of the sequents $\vdash (s/x)\phi_1$ and $\Gamma, EG_x(\phi_1)(s)\vdash EG_x(\phi_1)(s')$, respectively, then $|\pi|$ is the sequence $s|\pi_1|$.
		\end{itemize}
		The sequent $|\pi| = s_0,s_1,...,s_n$ is such that $s_0=s$; for all $i$ between $0$ and $n-1$, $s_i\longrightarrow s_{i+1}$; for all $i$ between $0$ and $n$, the sequent $\vdash (s_i/x)\phi_1$ has a proof smaller than $\pi$; and $s_n$ is equal to $s_p$ for some $p$ between $0$ and $n-1$. By induction hypothesis, for all $i$, we have $\models (s_i/x)\phi_1$. Using Proposition \ref{prop:fis}, there exists an infinite sequence $s_0',s_1',...$ such that for all $i$, we have $s_i'\longrightarrow s_{i+1}'$, and $\models (s_i'/x)\phi_1$. Hence, $\models EG_x(\phi_1)(s)$.
		\item If the last rule of $\pi$ is $\mathbf{AR}$-$\mathsf{R_1}$ or $\mathbf{AR}$-$\mathsf{R_2}$, then the proved sequent has the form $\vdash AR_x(\phi_1,\phi_2)(s)$. We associate a finite tree $|\pi|$ to the proof $\pi$ by induction in the following way.
		\begin{itemize}
			\item If the proof $\pi$ ends with the $\mathbf{AR}$-$\mathsf{R_1}$ rule with subproofs $\rho_1$ and $\rho_2$ of the sequents $\vdash (s/x)\phi_1$  and $\vdash (s/x)\phi_2$, respectively, or with the
			$\mathbf{AR}$-merge rule, then the tree contains a single node $s$.
			\item If the proof $\pi$ ends with the $\mathbf{AR}$-$\mathsf{R_2}$ rule, with subproofs $\rho, \pi_1,...,\pi_n$ of the sequents $\vdash (s/y)\phi_2$, $\Gamma, AR_{x,y}(\phi_1, \phi_2)(s) \vdash AR_{x,y}(\phi_1, \phi_2)(s_1),...,$ \\$\Gamma, AR_{x,y}(\phi_1, \phi_2)(s)\vdash AR_{x,y}(\phi_1, \phi_2)(s_n)$, respectively, then $|\pi|$ is the tree $s(|\pi_1|, ..., |\pi_n|)$.
		\end{itemize}
		The tree $|\pi|$ has root $s$; for each internal node $s'$, the children of this node are labeled by the elements of 
		$\textsf{Next}(s')$; for each node $s'$ of $|\pi|$, the sequent $\vdash (s'/y)\phi_2$ has a proof smaller than $\pi$; and for each leaf $s'$, either the sequent $\vdash (s'/x)\phi_1$ has a proof smaller than $\pi$, or $s'$ is also a label of a node on the branch from the root of $|\pi|$ to this leaf. By induction hypothesis, for each node $s'$ of this tree $\models (s'/y)\phi_2$ and for each leaf $s'$, either $\models (s'/x)\phi_1$ or $s'$ is also a label of a node on the branch from the root of $|\pi|$ to this leaf. Using Proposition \ref{prop:fit}, there exists a
possibly
infinite tree $T'$ labeled by states such that for each internal node $s$ the successors of $s$ are the elements of $\textsf{Next}(s)$, for each node $s'$ of $T'$, $\models (s'/y)\phi_2$, and for each leaf $s'$ of $T'$, $\models (s'/x)\phi_1$. Thus, $\models AR_{x,y}(\phi_1, \phi_2)(s)$.
		\item If the last rule of $\pi$ is $\mathbf{EU}$-$\mathsf{R_1}$ or $\mathbf{EU}$-$\mathsf{R_2}$, then the proved sequent has the form $\vdash EU_{x,y}(\phi_1, \phi_2)(s)$. We associate a finite sequence $|\pi|$ to the proof $\pi$ by induction in the following way.
		\begin{itemize}
			\item If the proof $\pi$ ends with the $\mathbf{EU}$-$\mathsf{R_1}$ rule with a subproof $\rho$ of the sequent $\vdash (s/y)\phi_2$, then the sequence contains a single element $s$.
			\item If the proof $\pi$ ends with the $\mathbf{EU}$-$\mathsf{R_2}$ rule, with subproofs $\rho$ and $\pi_1$ of the sequents $\vdash (s/x)\phi_1$ and $\vdash EU_{x,y}(\phi_1, \phi_2)(s')$, respectively, then $|\pi|$ is the sequence $s|\pi_1|$.
		\end{itemize}
		The sequence $|\pi| = s_0, ..., s_n$ is such that $s_0 = s$; for each $i$ between $0$ and $n-1$, $s_i \longrightarrow s_{i+1}$; for each $i$ between $0$ and $n - 1$, the sequent $\vdash (s_i/x)\phi_1$ has a proof smaller than $\pi$; and the sequent $\vdash (s_n/y)\phi_2$ has a proof smaller than $\pi$. By induction hypothesis, for each $i$ between $0$ and $n-1$, $\models (s_i/x)\phi_1$ and $\models (s_n/y)\phi_2$. Hence, $\models EU_{x,y}(\phi_1, \phi_2)(s)$.
		\item The last rule cannot be a merge rule.
	\end{itemize}
\end{proof}

\begin{theorem}[Completeness]\label{thm:complete}
	Let $\phi$ be a closed formula. If $\models \phi$ then the sequent
	$\vdash \phi$ is provable.
\end{theorem}
\begin{proof}
	By induction over the size of $\phi$.
	\begin{itemize}
		\item If $\phi = P(s_1,...,s_n)$, then as $\models P(s_1,...,s_n)$, the sequent $\vdash P(s_1,...,s_n)$ is provable with the rule \textsf{atom-R}.
		\item If $\phi = \neg P(s_1,...,s_n)$, then as $\models \neg P(s_1,...,s_n)$, the sequent $\vdash \neg P(s_1,...,s_n)$ is provable with the rule $\neg$-\textsf{R}.
		\item If $\phi = \top$, then $\vdash \top$ is provable with the rule $\top$-\textsf{R}.
		\item If $\phi = \bot$, then it is not the case that $\models \bot$.
		\item If $\phi = \phi_1\wedge\phi_2$, then as $\models \phi_1\wedge\phi_2$, $\models \phi_1$ and $\models \phi_2$. By induction hypothesis, the sequents $\vdash \phi_1$ and $\vdash \phi_2$ are provable. Thus the sequent $\vdash \phi_1\wedge\phi_2$ is provable with the $\wedge$-\textsf{R} rule.
		\item If $\phi = \phi_1\vee\phi_2$, as $\models \phi_1\vee\phi_2$, $\models \phi_1$ or $\models \phi_2$. By induction hypothesis, the sequent $\vdash \phi_1$ or $\vdash \phi_2$ is provable and the sequent $\vdash \phi_1\vee\phi_2$ is provable with the $\vee$-$\mathsf{R_1}$ or $\vee$-$\mathsf{R_2}$ rule, respectively.
		\item If $\phi = AX_x(\phi_1)(s)$, as $\models AX_x(\phi_1)(s)$, for each state $s'$ in $\textsf{Next}(s)$, we have $\models (s'/x)\phi_1$. By induction hypothesis, for each $s'$ in $\textsf{Next}(s)$, the sequent $\vdash (s'/x)\phi_1$ is provable. Using these proofs and the $\mathbf{AX}$-\textsf{R} rule, we build a proof of the sequent $\vdash AX_x(\phi_1)(s)$.
		\item If $\phi = EX_x(\phi_1)(s)$, as $\models EX_x(\phi_1)(s)$, there exists a state $s'$ in $\textsf{Next}(s)$ such that $\models (s'/x)\phi_1$. By induction hypothesis, the sequent $\vdash (s'/x)\phi_1$ is
		provable. With this proof and the $\mathbf{EX}$-\textsf{R} rule, we build a proof of the sequent $\vdash EX_x(\phi_1)(s)$.
		\item If $\phi = AF_x(\phi_1)(s)$, as $\models AF_x(\phi_1)(s)$, there exists a finite tree $T$ such that $T$ has root $s$, for each internal node $s'$, the children of this node are labeled by the elements of $\textsf{Next}(s')$, and for each leaf $s'$, $\models (s'/x)\phi_1$. By induction hypothesis, for every leaf $s'$, the sequent $\vdash (s'/x)\phi_1$ is provable. Then, to each subtree $T'$ of $T$, we associate a proof $|T'|$ of the sequent $\vdash AF_x(\phi_1)(s')$ where $s'$ is the root of $T'$, by induction, as follows.
		\begin{itemize}
			\item If $T'$ contains a single node $s'$, then the proof $|T|$ is built with the $\mathbf{AF}$-$\mathsf{R_1}$ rule from the proof of $\vdash (a/x)\phi_1$ given by the induction hypothesis.
			\item If $T' = s'(T_1, ..., T_n)$, then the proof $|T|$ is built with the $\mathbf{AF}$-$\mathsf{R_2}$ rule from the proofs $|T_1|, ..., |T_n|$ of the sequents $\vdash AF_x(\phi_1)(s_1), ..., \vdash AF_x(\phi_1)(s_n)$, respectively, where $s_1, ..., s_n$ are the elements of $\textsf{Next}(s')$.
		\end{itemize}
		
		This way, the proof $|T|$ is a proof of the sequent $\vdash AF_x(\phi_1)(s)$.
		\item If $\phi = EG_x(\phi_1)(s)$, as $\models EG_x(\phi_1)(s)$, there exists a path $s_0, s_1, ...$ such that $s_0 = s$ and for all $i$, $\models (s_i/x)\phi_1$. By induction hypothesis, all the sequents $\vdash (s_i/x)\phi_1$ are provable. Using Proposition \ref{prop:ifs}, there exists a finite sequence $T = s_0, ..., s_n$ such that for all $i$, $s_i \longrightarrow s_{i+1}$, the sequent $\vdash (s_i/x)\phi_1$ is provable and $s_n$ is some $s_p$ for $p < n$. We associate a proof $|s_i, ..., s_n|$ of the sequent $EG_x(\phi_1)(s_0),\ldots, EG_x(\phi_1)(s_{i-1}) \vdash EG_x(\phi_1)(s_i)$ to each suffix of $T$ by induction as follows.
		\begin{itemize}
			\item The proof $|s_n|$ is built with the $\mathbf{EG}$-merge rule.
			\item If $i \le n-1$, then the proof $|s_i, ..., s_n|$ is built with the $\mathbf{EG}$-\textsf{R} rule from the proof of $\vdash (s_i/x)\phi_1$ given by the induction hypothesis and the proof $|s_{i+1}, ..., s_n|$ of the sequent $EG_x(\phi_1)(s_0), ..., EG_x(\phi_1)(s_i) \vdash EG_x(\phi_1)(s_{i+1})$.
		\end{itemize}
		
		This way, the proof $|s_0,...,s_n|$ is a proof of the sequent $\vdash EG_x(\phi_1)(s)$.

              \item If $\phi = AR_{x,y}(\phi_1, \phi_2)(s)$, as
                $\models AR_{x,y}(\phi_1, \phi_2)(s)$, there exists an
                possibly
                infinite tree such that the root of this tree is $s$,
                for each internal node $s'$, the children of this node
                are labeled by the elements of $\textsf{Next}(s')$,
                for each node $s'$, $\models (s'/y)\phi_2$ and for
                each leaf $s'$, $\models (s'/x)\phi_1$. By induction
                hypothesis, for each node $s'$ of the tree, the
                sequent $\vdash (s'/y)\phi_2$ is provable and for each
                leaf $s'$ of the tree, the sequent
                $\vdash (s'/x)\phi_1$ is provable. Using Proposition
                \ref{prop:ift}, there exists a finite tree $T$ such
                that for each internal node $s'$ the successors of
                $s'$ are the elements of $\textsf{Next}(s')$, for each
                node $s'$, the sequent $\vdash (s'/y)\phi_2$ is
                provable, and for each leaf $s'$, either the sequent
                $\vdash (s'/x)\phi_1$ is provable or $s'$ is also a
                label of a node on the branch from the root of $T$ to
                this leaf. Then, to each subtree $T'$ of $T$, we
                associate a proof $|T'|$ of the sequent
                $AR_{x,y}(\phi_1, \phi_2)(s_1), ..., AR_{x,y}(\phi_1,
                \phi_2)(s_m) \vdash AR_{x,y}(\phi_1, \phi_2)(s')$
                where $s'$ is the root of $T'$ and $s_1, ..., s_m$ is
                the sequence of nodes in $T$ from the root of $T$ to
                the root of $T'$.
		
		\begin{itemize}
			\item If $T'$ contains a single node $s'$, and the sequent $\vdash (s'/x)\phi_1$ is provable then the proof $|T'|$ is built with the $\mathbf{AR}$-$\mathsf{R_1}$ rule from the proofs of $\vdash (s'/x)\phi_1$ and $\vdash (s'/y)\phi_2$ given by the induction hypothesis.
			\item If $T'$ contains a single node $s'$, and $s'$ is among $s_1,...,s_m$, then the proof $|T'|$ is built with the $\mathbf{AR}$-merge rule.
			\item If $T' = s'(T_1, ..., T_n)$, then the proof $|T'|$ is built with the $\mathbf{AR}$-$\mathsf{R_2}$ rule from the proofs $\vdash (s'/y)\phi_2$ given by the induction hypothesis and the proofs $|T_1|, ..., |T_n|$ of the sequents
			\begin{center}
				$AR_{x,y}(\phi_1,\phi_2)(s_1),...,AR_{x,y}(\phi_1,\phi_2)(s_m)$, $AR_{x,y}(\phi_1,\phi_2)(s')\vdash AR_{x, y}(\phi_1,\phi_2)(s_1')$\\
				...\\
				$AR_{x,y}(\phi_1,\phi_2)(s_1),...,AR_{x,y}(\phi_1,\phi_2)(s_m)$, $AR_{x,y}(\phi_1,\phi_2)(s')\vdash AR_{x, y}(\phi_1,\phi_2)(s_n')$
			\end{center}
			respectively, where $s_1',...,s_n'$ are the elements of $\textsf{Next}(s')$.
		\end{itemize}
		
		This way, the proof $|T|$ is a proof of the sequent $\vdash AR_{x,y}(\phi_1,\phi_2)(s)$.
		\item If $\phi = EU_{x,y}(\phi_1, \phi_2)(s)$, as $\models EU_{x,y}(\phi_1, \phi_2)(s)$, there exists a finite sequence $T = s_0, ..., s_n$ such that $\models (s_n/y)\phi_2$ and for all $i$ between $0$ and $n-1$, $\models (s_i/x)\phi_1$. By induction hypothesis, the sequent $\vdash (s_n/y)\phi_2$ is provable and for all $i$ between $0$ and $n-1$, the sequent $\vdash (s_i/x)\phi_1$ is provable. We associate a proof $|s_i, ..., s_n|$ of the sequent $\vdash EU_{x,y}(\phi_1, \phi_2)(s_i)$ to each suffix of $T$ by induction as follows.
		
		\begin{itemize}
			\item The proof $|s_n|$ is built with the $\mathbf{EG}$-$\mathsf{R_1}$ rule from the proof of $\vdash (s_n/y)\phi_2$ given by the induction hypothesis.
			\item If $i\le n-1$, then the proof $|s_i,...,s_n|$ is built with the $\mathbf{EG}$-$\mathsf{R_2}$ rule from the proof of $\vdash (s_i/x)\phi_1$ given by the induction hypothesis and the proof $|s_{i+1},...,s_n|$ of the sequent $\vdash EU_{x,y}(\phi_1,\phi_2)(s_{i+1})$.
		\end{itemize}
		This way, the proof $|s_0,...,s_n|$ is a proof of the sequent $\vdash EU_{x,y}(\phi_1,\phi_2)(s)$.
	\end{itemize}
	
\end{proof}

\section{Proof of the correctness of the proof search algorithm}\label{appendix:proof:corret:prfsearch}

\begin{proposition}
Given a formula $\phi$, 
$\textsf{cpt}(\vdash\phi, \mathfrak{t}, \mathfrak{f}) \rsa^* \mathfrak{t}$ iff $\vdash \phi$ is provable.
\end{proposition}

\begin{proof}
We prove, more generally, by induction on the structure of $\phi$, that 
given a sequent $\Gamma \vdash \phi$ and distinct \textsf{CPT}s
        $c_1$ and $c_2$, $\textsf{cpt}(\Gamma \vdash\phi, c_1, c_2) 
        \rsa^* c_1$ iff $\Gamma \vdash \phi$ is provable.
	\begin{itemize}
		\item If $\phi=\top$ or $\bot$, trivial.
		\item If $\phi=P(s_1,...,s_n)$ where $P(s_1,...,s_n)$ is atomic, then
		$\textsf{cpt}(\vdash P(s_1,...,s_n), c_1, c_2)\rsa c_1$ iff $\langle
		s_1,s_2,...,s_n\rangle \in P$ iff $\vdash P(s_1,s_2,...,s_n)$ is
		provable.
		\item If $\phi=\neg P(s_1,...,s_n)$ where $P(s_1,...,s_n)$ is atomic,
		then $\textsf{cpt}(\vdash \neg P(s_1,...,s_n), c_1, c_2)\rsa c_1$ iff
		$\langle s_1,s_2,...,s_n\rangle \notin P$ iff $\vdash \neg
		P(s_1,s_2,...,s_n)$ is provable.
		\item If $\phi = \phi_1\wedge \phi_2$, then
		$\textsf{cpt}(\vdash\phi_1\wedge\phi_2, c_1, c_2)\rsa^* c_1$ iff $\textsf{cpt}(\vdash
		\phi_1\wedge\phi_2, c_1, c_2)\rsa \textsf{cpt}(\vdash \phi_1, \textsf{cpt}(\vdash
		\phi_2, c_1, c_2), c_2) \rsa^* \textsf{cpt}(\vdash \phi_2, c_1, c_2)\rsa^*
		c_1$ iff both $\vdash \phi_1$ and $\vdash \phi_2$ are provable (by
		induction hypothesis) iff $\vdash \phi_1\wedge \phi_2$ are provable.
		\item If $\phi = \phi_1\vee\phi_2$, then $\textsf{cpt}(\vdash \phi_1\vee\phi_2,
		c_1, c_2)\rsa^* c_1$ iff either $\textsf{cpt}(\vdash\phi_1\vee\phi_2, c_1,
		c_2)\rsa \textsf{cpt}(\vdash \phi_1, c_1, \textsf{cpt}(\vdash\phi_2, c_1, c_2))\rsa^*
		c_1$, or $\textsf{cpt}(\vdash\phi_1\vee\phi_2, c_1, c_2)\rsa \textsf{cpt}(\vdash
		\phi_1, c_1, \textsf{cpt}(\vdash\phi_2, c_1, c_2))\rsa^* \textsf{cpt}(\vdash\phi_2,
		c_1, c_2) \rsa^* c_1$ iff either $\vdash \phi_1$, or $\vdash \phi_2$
		is provable (by induction hypothesis) iff $\vdash \phi_1\vee\phi_2$
		is provable.
		\item If $\phi = AX_x(\psi)(s)$ and $\{s_1,...,s_n\}=\textsf{Next}(s)$, then
		$\textsf{cpt}(\vdash AX_x(\psi)(s), c_1, c_2)\rsa^* c_1$ iff $\textsf{cpt}(\vdash
		AX_x(\psi)(s), c_1, c_2)\rsa \\\textsf{cpt}(\vdash(s_1/x)\psi,
		\textsf{cpt}(\vdash(s_2/x)\psi, \textsf{cpt}(...\textsf{cpt}(\vdash(s_n/x)\psi, c_1,
		c_2)...,c_2), c_2), c_2)\rsa^*\\ \textsf{cpt}(\vdash(s_2/x)\psi,
		\textsf{cpt}(...\textsf{cpt}(\vdash(s_n/x)\psi, c_1, c_2)...,c_2), c_2)\rsa^*
		...\rsa^* \\\textsf{cpt}(\vdash (s_n/x)\psi, c_1, c_2)\rsa^* c_1$ iff $\vdash
		(s_1/x)\psi, \vdash (s_2/x)\psi, ..., \vdash (s_n/x)\psi$ are all
		provable (by induction hypothesis) iff $\vdash AX_x(\psi)(s)$ is
		provable.
		\item If $\phi = EX_x(\psi)(s)$ and $\{s_1,...,s_n\}=\textsf{Next}(s)$, then
		$\textsf{cpt}(\vdash EX_x(\psi)(s), c_1, c_2)\rsa^* c_1$ iff \\ $\textsf{cpt}(\vdash
		EX_x(\psi)(s), c_1, c_2)\rsa \\ \textsf{cpt}(\vdash (s_1/x)\psi, c_1,
		\textsf{cpt}(\vdash(s_2/x)\psi, c_1, \textsf{cpt}(...\textsf{cpt}(\vdash(s_n/x)\psi, c_1,
		c_2)...)))\rsa^* \\ \textsf{cpt}(\vdash(s_i/x)\psi, c_1, c_2')\rsa^* c_1$ iff
		$\vdash (s_i/x)\psi$ is provable (by induction hypothesis) iff
		$\vdash EX_x(\psi)(s)$ is provable, where $1\le i \le n$, and $c_2'$
		is either $c_2$ when $i=n$ or $\textsf{cpt}(\vdash(s_{i+1}/x)\psi, c_1,
		\textsf{cpt}(...\textsf{cpt}(\vdash(s_n/x)\psi, c_1, c_2)...))$ when $i\neq n$.
		\item If $\phi = AF_x(\psi)(s)$ and $\{s_1,...,s_n\}=\textsf{Next}(s)$, then
		$\textsf{cpt}(\Gamma\vdash AF_x(\psi)(s), c_1, c_2)\rsa^* c_1$ iff either
		\\ $\textsf{cpt}(\Gamma\vdash AF_x(\psi)(s), c_1, c_2)\rsa \\ \textsf{cpt}(\vdash
		(s/x)\psi, c_1, \textsf{cpt}(\Gamma'\vdash AF_x(\psi)(s_1),
		\textsf{cpt}(...\textsf{cpt}(\\ \Gamma'\vdash AF_x(\psi)(s_n), c_1, c_2)..., c_2),
		c_2))\rsa^* c_1$, or \\ $\textsf{cpt}(\Gamma\vdash AF_x(\psi)(s), c_1,
		c_2)\rsa \\ \textsf{cpt}(\vdash(s/x)\psi, c_1, \textsf{cpt}(\Gamma'\vdash
		AF_x(\psi)(s_1), \textsf{cpt}(...\textsf{cpt}(\\ \Gamma'\vdash AF_x(\psi)(s_n), c_1,
		c_2)..., c_2), c_2))\rsa^* \\ \textsf{cpt}(\Gamma'\vdash AF_x(\psi)(s_i),
		\textsf{cpt}(...\textsf{cpt}(\Gamma'\vdash AF_x(\psi)(s_n), c_1, c_2)...,c_2),
		c_2)\rsa^* c_1$, where $1\le i\le n$ and
		$\Gamma'=\Gamma,AF_x(\psi)(s)$. 
We are going to prove that 
the second condition holds iff $\Gamma, AF_x(\psi)(s)\vdash
		AF_x(\psi)(s_1)$, $\Gamma , AF_x(\psi)(s)\vdash AF_x(\psi)(s_2)$,
		$...$, $\Gamma , AF_x(\psi)(s)\vdash AF_x(\psi)(s_n)$ are all
		provable.
This will be sufficient to conclude as 
The first condition holds iff
		$\vdash (s/x)\psi$ is provable (by induction hypothesis). 
So, we will get that both conditions hold if and only if 
$\Gamma\vdash
		AF_x(\psi)(s)$ is provable. 
Let us prove, as announced that 
the	second condition holds iff $\Gamma, AF_x(\psi)(s)\vdash
		AF_x(\psi)(s_1)$, $\Gamma , AF_x(\psi)(s)\vdash AF_x(\psi)(s_2)$,
		$...$, $\Gamma , AF_x(\psi)(s)\vdash AF_x(\psi)(s_n)$ are all
		provable:
		\begin{itemize}
			\item ($\A$) if the second condition holds, then $\Gamma,
			AF_x(\psi)(s)\vdash AF_x(\psi)(s_1)$, $\Gamma ,
			AF_x(\psi)(s)\vdash AF_x(\psi)(s_2)$, $...$, $\Gamma ,
			AF_x(\psi)(s)\vdash AF_x(\psi)(s_n)$ are all provable. That is
			because otherwise, if $1\le j\le n$ such that
			$\Gamma,AF_x(\psi)(s)\vdash AF_x(\psi)(s_j)$ is the first sequent
			that is not provable, then there exists an infinite path $s_{j_0},
			s_{j_1},s_{j_2}, ...$ and $s_{j_0}=s_j$ such that $\vdash
			(s_{j_k}/x)\psi$ is not provable for all $k \ge 0$, then by
			induction hypothesis, \\$\textsf{cpt}(\Gamma'\vdash AF_x(\psi)(s_1),
			\textsf{cpt}(...\textsf{cpt}(\Gamma'\vdash AF_x(\psi)(s_n), c_1, c_2)..., c_2), c_2)
			\rsa^*\\ \textsf{cpt}(\Gamma'\vdash AF_x(\psi)(s_{j_0}), \textsf{cpt}(\Gamma'\vdash
			AF_x(\psi)(s_{j+1}), \textsf{cpt}(...\textsf{cpt}( \Gamma'\vdash AF_x(\psi)(s_n),\\
			c_1, c_2)..., c_2), c_2), c_2) \rsa^*\\ \textsf{cpt}(\Gamma_1\vdash
			AF_x(\psi)(s_{j_1}), c_1^1, c_2)\rsa^* ... \rsa^*
			\textsf{cpt}(\Gamma_m\vdash AF_x(\psi)(s_{j_m}), c_1^m, c_2)\rsa^*
			c_2$, where \\$AF_x(\psi)(s_{j_m})\in \Gamma_m$, \\$\Gamma' =
			\Gamma, AF_x(\psi)(s)$, \\$\Gamma_m = \Gamma',
			AF_x(\psi)(s_{j_0}),...,AF_x(\psi)(s_{j_m})$, and the shape of
			$c_1^1,...,c_1^m$ have no impact on the transformations of \textsf{CPT}s
			here. Note that such $m \ge 0$ exists because our Kripke model is
			finite. So, the second condition holds implies that $\Gamma,
			AF_x(\psi)(s)\vdash AF_x(\psi)(s_1), \Gamma, AF_x(\psi)(s)\vdash
			AF_x(\psi)(s_2)$, $...$, $\Gamma, AF_x(\psi)(s)\vdash
			AF_x(\psi)(s_n)$ are all provable, and thus $\Gamma \vdash
			AF_x(\psi)(s)$ are provable.
			\item ($\Leftarrow$) if $\forall i\in \{1,2,...,n\}$, $\Gamma,
			AF_x(\psi)(s)\vdash AF_x(\psi)(s_i)$ is provable, to prove that
			the second condition holds, it is sufficient to prove that\\
			$\textsf{cpt}(\Gamma'\vdash AF_x(\psi)(s_j), c_1', c_2')\rsa^* c_1'$ for
			all $1\le j\le n$ and all $c_1'$, $c_2'$, and that
			$\Gamma'=\Gamma, AF_x(\psi)(s)$. This is easily proved by
			induction on the structure of the proof tree of $\Gamma,
			AF_x(\psi)(s)\vdash AF_x(\psi)(s_j)$.
		\end{itemize}
		\item if $\phi = EG_x(\psi)(s)$ and $\{s_1,...,s_n\}=\textsf{Next}(s)$, then
		for $\Gamma \vdash EG_x(\psi)(s)$,
		\begin{itemize}
			\item if $EG_x(\psi)(s)\in \Gamma$, trivial;
			\item if $EG_x(\psi)(s)\notin \Gamma$, then \\ $\textsf{cpt}(\Gamma\vdash
			EG_x(\phi)(s), c_1, c_2)\rsa^* c_1$ iff \\ $\textsf{cpt}(\Gamma\vdash
			EG_x(\psi)(s), c_1, c_2)\rsa \\ \textsf{cpt}(\vdash(s/x)\psi,
			\textsf{cpt}(\Gamma'\vdash EG_x(\psi)(s_1), c_1,\\
			\textsf{cpt}(...\textsf{cpt}( \Gamma'\vdash EG_x(\psi)(s_n), c_1, c_2)...)),
			c_2)\rsa^*\\  \textsf{cpt}(\Gamma'\vdash EG_x(\psi)(s_{i_1}), c_1,
			\textsf{cpt}(...\textsf{cpt}(\Gamma'\vdash EG_x(\psi)(s_n), c_1, c_2)...))\rsa^* \\
			\textsf{cpt}(\Gamma_1\vdash EG_x(\psi)(s_{i_11}), c_1,
			\textsf{cpt}(...\textsf{cpt}(\Gamma_1\vdash EG_x(\psi)(s_{i_1n_1}), c_1,
			c_2^1)...))\rsa^* \\ \textsf{cpt}(\Gamma_2\vdash EG_x(\psi)(s_{i_1i_21}),
			c_1, \textsf{cpt}(...\textsf{cpt}(\Gamma_2\vdash EG_x(\psi)(s_{i_1i_2n_2}), c_1,
			c_2^2)...))\rsa^* \\ \textsf{cpt}(\Gamma_m\vdash
			EG_x(\psi)(s_{i_1i_2...i_m1}), c_1, \textsf{cpt}(...\textsf{cpt}( \Gamma_m\vdash
			EG_x(\psi)(s_{i_1i_2...i_mn_m}),\\ c_1, c_2^m)...))\rsa^*\\
			\textsf{cpt}(\Gamma_m\vdash EG_x(\psi)(s_{i_1i_2...i_mi_{m+1}}), c_1,
			\textsf{cpt}(...\textsf{cpt}( \Gamma_m\vdash EG_x(\psi)(s_{i_1i_2...i_mn_m}),\\ c_1,
			c_2^m)...)) \rsa^* c_1$ iff there exists an infinite path
			$$s,s_{i_1},s_{i_1i_2},...,s_{i_1i_2...i_m},s_{i_1i_2...i_mi_{m+1}},...$$
			such that for all state $s'$ in this path, $\vdash \psi(s')$ is
			provable, where $\Gamma' = \Gamma, EG_x(\psi)(s)$, $\Gamma_m =
			\Gamma', EG_x(\psi)(s_{i_1}),..., EG_x(\psi)(s_{i_1i_2...i_m})$,
			and $EG_x(\psi)(s_{i_1i_2...i_mi_{m+1}}) \in \Gamma_m$. By
			induction hypothesis, this holds iff $\Gamma \vdash EG_x(\psi)(s)$
			is provable.
			\item if $\phi = AR_x,y(\phi_1,\phi_2)(s)$, as are both co-inductive
			modalities, the analysis is analogous to $EG$.
			\item if $\phi = EU_x,y(\phi_1,\phi_2)(s)$, as are both inductive
			modalities, the analysis is analogous to $AF$.
	\end{itemize}
	\end{itemize}
\end{proof}



\end{document}